\documentclass[
  twocolumn,
  english,
  aps,
  prapplied,
  longbibliography,
  superscriptaddress,
  amsmath,
  amssymb,
  floatfix,
]{revtex4-2}

\usepackage{graphicx}
\usepackage{dcolumn}
\usepackage{bm}

\usepackage{graphicx}
\usepackage{amsmath}
\usepackage{tabularx}
\usepackage{array}
\usepackage{booktabs}
\usepackage{makecell}
\usepackage{multirow}

\usepackage{float}
\makeatletter
\let\newfloat\newfloat@ltx
\makeatother
\usepackage{algorithm}
\usepackage{algpseudocode}

\usepackage{amsthm}
\newtheorem{proposition}{Proposition}

\usepackage{braket}

\begin{document}

\preprint{AAPM/123-QED}


\title{Towards Natural Gas Contract Selection\\via Quantum-Guided Independent Set Reduction}

\author{Vivek Dixit}
\email{vivekdixit@ibm.com}
\affiliation{IBM Research}

\author{Vaibhaw Kumar}
\affiliation{IBM Research}

\author{Kentaro Ohno}
\affiliation{IBM Research}

\author{Alberto Maldonado Romo}
\affiliation{IBM Research}

\author{Larry Bowden}
\email{larry.bowden@woodside.com}
\affiliation{Woodside Energy}

\date{\today}

\begin{abstract}
Selecting mutually compatible natural gas transportation contracts is a practically important optimization task in which operators must choose from many candidate agreements subject to temporal, infrastructural, and flow-related constraints. As the number of candidates grows, the resulting search space becomes difficult to explore exhaustively. We study a pairwise abstraction of this task, formulated as a Maximum Clique problem on a contract-compatibility graph, or equivalently as a Maximum Independent Set (MIS) problem on the complement graph.

Building on recent work, this paper studies a quantum--classical framework for solving large-scale MIS instances within the limitations of noisy quantum hardware. The approach combines iterative classical graph reduction with quantum-guided optimization to progressively simplify the search space while maintaining high solution quality. This enables large candidate spaces to be reduced to smaller subproblems that are more suitable for execution on current quantum computers.

We evaluate the approach on fifteen benchmark instances from the \textit{Quantum Optimization Benchmarking Library} (QOBLIB), obtaining an average approximation ratio of 0.996 and recovering optimal solutions for fourteen instances, including graphs with up to 186 vertices.
We further evaluate the algorithm on six synthetic pairwise contract-compatibility graphs containing up to 900 contracts, where the proposed method achieves an average approximation ratio of 0.989 and obtains optimal solutions in four cases. These experiments demonstrate the ability of the hybrid MIS solver to reduce industrially motivated graphs. The pairwise abstraction serves as the first step of a two-stage screening procedure that narrows the candidate contracts to a smaller set of mutually compatible ones, which can then be verified against pipeline-capacity constraints.
\end{abstract}

\maketitle

\section{Introduction}\label{sec:introduction}

In large-scale energy transportation systems, efficiently allocating limited pipeline capacity is an important operational and economic challenge.
Selecting mutually compatible natural gas transportation contracts is a practically important optimization task in which operators must choose from many candidate agreements subject to temporal, infrastructural, and flow-related constraints. As the number of candidate contracts grows, the search space becomes increasingly difficult to explore exhaustively. In this work, we consider a pairwise abstraction of this task in which contracts are represented as vertices of a compatibility graph and pairwise compatibility relationships are represented as edges. Solving the Maximum Clique problem on this compatibility graph 
yields the largest subset of mutually compatible contracts. 


The Maximum Independent Set (MIS) problem is a fundamental combinatorial optimization problem defined on an undirected graph $G=(V,E)$, where the objective is to identify the largest subset $S\subseteq V$ such that no two vertices in $S$ are adjacent~\cite{Garey1979}. The MIS problem is closely related to the Maximum Clique through graph complementation: a maximum clique in a graph $G$ corresponds exactly to a maximum independent set in its complement graph $\bar{G}$~\cite{Bomze1999}. As one of Karp’s original NP-complete problems~\cite{Karp1972}, MIS has no known polynomial-time exact algorithm unless $\mathrm{P}=\mathrm{NP}$, and remains notoriously difficult to solve and approximate in the worst case~\cite{Hastad1996}. Consequently, exact algorithms exhibit exponential scaling, while heuristic methods may return sub-optimal solutions and lack optimality guarantees. The computational challenges become particularly severe for large and dense graphs encountered in practical applications, including finance \cite{Boginski2006, yalovetzky2026quantum}, wireless network scheduling~\cite{Joo2009}, bioinformatics~\cite{Halu2014}, compiler register allocation~\cite{Chaitin1982}, and computer vision feature matching~\cite{Leordeanu2005}. The inherent computational complexity of MIS has therefore motivated investigating alternative computing paradigms, including hybrid quantum--classical optimization algorithms~\cite{Farhi2014,Harrigan2021,schuetz2025qredumis, yalovetzky2026quantum}.

A variety of exact algorithms has been developed for MIS, each
balancing generality against computational efficiency. Branch-and-bound
methods systematically partition the solution space, pruning subproblems
whose upper bounds fall below the incumbent solution~\cite{Tarjan1977}.
The Bron--Kerbosch algorithm~\cite{BronKerbosch1973}, applicable to
MIS via graph complementation, remains a foundational reference, with
refinements by Tomita and Seki~\cite{Tomita2003} improving performance on sparse graphs. For structured graphs, dynamic programming yields
polynomial-time solutions on perfect graphs~\cite{Grotschel1984},
interval graphs~\cite{Golumbic1980}, and bounded-tree width
graphs~\cite{Bodlaender1993}. Integer Linear Programming (ILP)
formulations encode MIS as a binary optimization over vertex-selection
variables subject to edge constraints, and modern solvers such as CPLEX
and Gurobi with branch-and-cut strategies have extended solvable instance
sizes considerably~\cite{Nemhauser1975,Segundo2011}. More recently,
branch-and-reduce frameworks such as KaMIS combine exact kernelization
(data reduction) rules with branch-and-bound search, enabling exact
solutions on sparse real-world graphs with millions of vertices that are
intractable for classical branch-and-bound alone~\cite{Lamm2017,Hespe2019}.
Nevertheless, all exact methods share exponential worst-case complexity and can become prohibitive on dense, weakly structured graphs beyond a few
hundred vertices~\cite{Pardalos1994} --- precisely the regime
encountered in large-scale natural gas transportation networks.

Given the intractability of exact methods, heuristic and approximation
algorithms have been developed to obtain high-quality solutions within
acceptable time. Greedy algorithms iteratively select the minimum-degree
vertex and remove its neighbours, running in polynomial time but yielding
solutions far from optimal on dense or adversarially structured
graphs~\cite{Halldorsson1995,Alimonti1997}. Metaheuristics overcome
this by exploring broader regions of the solution space: Simulated
Annealing accepts worsening moves with a temperature-controlled
probability to escape local optima~\cite{Kirkpatrick1983}, while Tabu
Search maintains a memory structure that forbids recently visited
configurations~\cite{Glover1989}. Both have demonstrated substantially
improved quality over greedy initialization~\cite{Battiti1994}. For
massive sparse networks, ReduMIS intermixes polynomial-time
kernelization with an evolutionary local-search algorithm, iteratively
fixing likely-optimal low-degree vertices to reopen the reduction space,
and has been shown to find near-optimal or provably optimal solutions on
graphs with millions of vertices where classical metaheuristics
stall~\cite{Lamm2017}.
Semidefinite Programming
(SDP) relaxations represent the strongest known polynomial-time
approximation providing guarantees via the Lov\'{a}sz theta function
$\vartheta(G)$~\cite{Lovasz1979}, yet no polynomial-time algorithm can approximate MIS within
$n^{1-\epsilon}$ under standard assumptions~\cite{zuckerman_2006_inapproximability}. In
natural gas transportation scheduling, where solution quality directly
impacts revenue, the degraded performance of classical heuristics highlights challenges that motivate the investigation of alternative computational approaches.

Classical MIS methods face bottlenecks that limit their
applicability to large-scale, time-sensitive problems. 
Heuristics alleviate computational cost but introduce a quality--speed
trade-off: greedy methods sacrifice optimality, while metaheuristics
require careful tuning with no formal guarantee of solution
quality~\cite{Hoos2004} --- a critical limitation in revenue- and
safety-sensitive applications. 
Most fundamentally, local search methods
are susceptible to entrapment in local optima on the highly non-convex
MIS energy landscape~\cite{Andrade2012}. In natural gas
transportation scheduling, where the conflict graph can involve hundreds
to thousands of nodes and decisions must be made in near-real
time~\cite{Zheng2010}, these limitations motivate investigating alternative algorithms, including  quantum-classical methods. A key enabler is the Quadratic Unconstrained Binary Optimization (QUBO) formulation, which recasts the MIS problem as the minimization of a quadratic objective over binary variables~\cite{Lucas2014,Glover2019}, which is compatible with many quantum optimization algorithms, such as the Quantum Approximate Optimization Algorithm (QAOA)~\cite{Farhi2014}.  
QAOA employs parameterized variational circuits to generate high-quality approximate solutions and has been demonstrated for MIS on superconducting quantum processors~\cite{Harrigan2021,ebadi2022}. However, the limitations of noisy intermediate-scale quantum (NISQ) hardware~\cite{Preskill2018}, including restricted qubit counts, shallow circuit depths, and gate errors, currently prevent purely quantum methods from solving large industrial-scale optimization problems efficiently.

To overcome these limitations, quantum--classical algorithms combine the strengths of classical preprocessing with quantum-guided optimization. Schuetz \textit{et al.}~\cite{schuetz2025qredumis} proposed \textit{qReduMIS}, which integrates classical kernelization with quantum-derived information to identify vertices that are likely to belong to large independent sets, thereby enabling additional graph reductions. The method was validated on hardware-native random Union Jack graph instances using Rydberg quantum hardware, demonstrating improved performance relative to the other quantum optimization methods considered in that study. Similarly, Wybo \textit{et al.}~\cite{wybo2026} introduced a hybrid MIS algorithm that combines shallow-depth QAOA ($p=4$) with a classical greedy heuristic using pre-computed QAOA parameters. Their approach was implemented on a 20-qubit superconducting quantum processor and reported improved solution quality relative to classical greedy heuristics on the regular graph instances considered. More recently, Yalovetzky \textit{et al.}~\cite{yalovetzky2026quantum} proposed a quantum--classical framework for portfolio diversification, formulated as an MIS problem. Their method uses QAOA measurements to identify likely optimal vertices, enabling further optimal classical reductions rather than relying on QAOA to directly construct the final solution. Evaluated on real financial datasets using QAOA with depth $p=2$ on a trapped-ion processor, the algorithm solved quantum subproblems using up to 78 qubits, achieved average approximation ratios exceeding $0.96$, and outperformed standalone QAOA on the largest benchmark instances.

This work studies how quantum-informed recursive MIS reduction can be adapted
to superconducting quantum hardware and applied to industrially motivated
natural-gas contract-selection problems. Building on recent qReduMIS-style
ideas, we combine exact graph reductions, QAOA-derived marginal selection,
hardware-aware circuit synthesis, and hierarchical parameter optimization.
We test the algorithm on instances from the \emph{Quantum Optimization Benchmarking Library} (QOBLIB) \cite{koch2026qoblib} as well as synthetic instances of the considered contract scheduling problem, where it consistently finds near-optimal solutions.

Concretely, the framework converts the problem into a QUBO model, iteratively
applies classical MIS reduction rules to simplify the graph, and processes the
resulting irreducible kernel with the Quantum Approximate Optimization
Algorithm (QAOA), using the measured vertex-inclusion probabilities to steer
subsequent reductions until the graph is fully resolved. Relative to prior
hybrid MIS solvers, the specific contributions of this paper are as follows:

\begin{itemize}
\item \textbf{Domain adaptation to natural gas contract scheduling.} We cast
the contract-selection problem in natural gas pipeline transportation as a
Maximum Clique instance on a pairwise contract-compatibility graph and, via
graph complementation, solve it as an MIS problem with the proposed hybrid
solver. We formalize the compatibility model---temporal overlap, shared
infrastructure, pairwise capacity screening, priorities, and flow
conflicts---and clearly delimit its scope, showing where the pairwise
abstraction departs from a full pipeline schedule with cumulative
segment--period capacity constraints.

\item \textbf{Hardware-aware circuit synthesis.} We
introduce a connectivity-aware circuit synthesis that embeds the QUBO
interaction graph onto the qubit topology of an IBM superconducting processor given a limited SWAP budget. Combined with a hierarchical QAOA parameter schedule (depth-one
grid search, COBYLA refinement, and recursive depth extension), this
helps control the two-qubit depth and enables finding high-quality solutions of MIS kernels containing up to 125 qubits on 156-qubit hardware.

\item \textbf{Systematic benchmarking on QOBLIB.} We evaluate the solver on
fifteen MIS instances from the Quantum Optimization
Benchmarking Library (QOBLIB), spanning 17--186 vertices and 39--5{,}173
edges. Against CPLEX-verified optima the hybrid solver attains an average
approximation ratio of 0.996 and recovers the exact optimum on fourteen of
fifteen instances, providing a hardware-executed evaluation to the open benchmark suite.

\item \textbf{Empirical analysis of quantum marginals versus uniform
selection.} We isolate the contribution of the quantum stage using a matched
uniform-sampling baseline that differs only in how candidate vertices are
selected after each reduction step. We then quantify the advantage of guiding
reductions using the full QAOA measurement distribution rather than relying
solely on a single lowest-energy bitstring. For the \textit{C125-9} instance
from QOBLIB, the quantum-guided strategy concentrates substantially more
probability mass on high-quality solutions---for example, it is $12\times$
more likely to reach an independent set of size $31$ or larger and recovers
the optimum in cases where uniform sampling never does---providing
evidence that the marginal probabilities contain valuable information for
guiding reductions.

\end{itemize}

A word on scope is warranted before we proceed. At the problem sizes
that fit today's quantum hardware, classical exact MIS solvers---branch-and-reduce
and kernelization-based tools such as KaMIS/ReduMIS~\cite{Lamm2017}---are
both faster and exact, and we make no claim to surpass them. Our contribution
is instead methodological. The value lies in the pipeline itself---hardware-aware
embedding, the hierarchical QAOA schedule, and quantum-guided reduction---which
transfers to near-term quantum optimization more broadly; here MIS is the
measuring stick, not the product. We adopt it as a testbed precisely because
its known optima and strong classical reductions let us isolate and quantify
the quantum contribution, and because it admits a genuine hardware demonstration
with ground truth known at up to 125 qubits. The bet we make explicit is that
the quantum workload scales with qubit count rather than with problem size, so
the same architecture is positioned to ride future hardware into the dense
regimes where exact classical methods break down. This paper should therefore
be read as near-term quantum-optimization methods research validated on a
benchmarkable problem, not as a faster MIS solver.

The remainder of this paper is organized as follows.
Section~\ref{sec:methods} presents the methodology, including the QUBO and
Ising formulation, the classical reduction and decomposition stages, the
hardware-aware circuit synthesis, the QAOA parameter-optimization schedule,
and the quantum-guided reduction procedure together with its randomized
baseline. Section~\ref{sec:results} reports the experimental evaluation on the
QOBLIB benchmarks and the natural gas contract-compatibility graphs, including
detailed case studies, the statistical analysis of quantum marginals versus
random selection, and the runtime and run-to-run reproducibility of the solver. 
Section~\ref{sec:conclusion} concludes and discusses directions for future
work.

\section{Methodology}\label{sec:methods}

Within this section we formally introduce the considered problem and describe each stage of the algorithm in detail, from the contract-compatibility graph through the iterative reduce--QAOA--select loop to the final independent set, as summarized in Figure~\ref{fig:pipeline}.

\begin{figure*}[t]
\centering
\includegraphics[width=\textwidth]{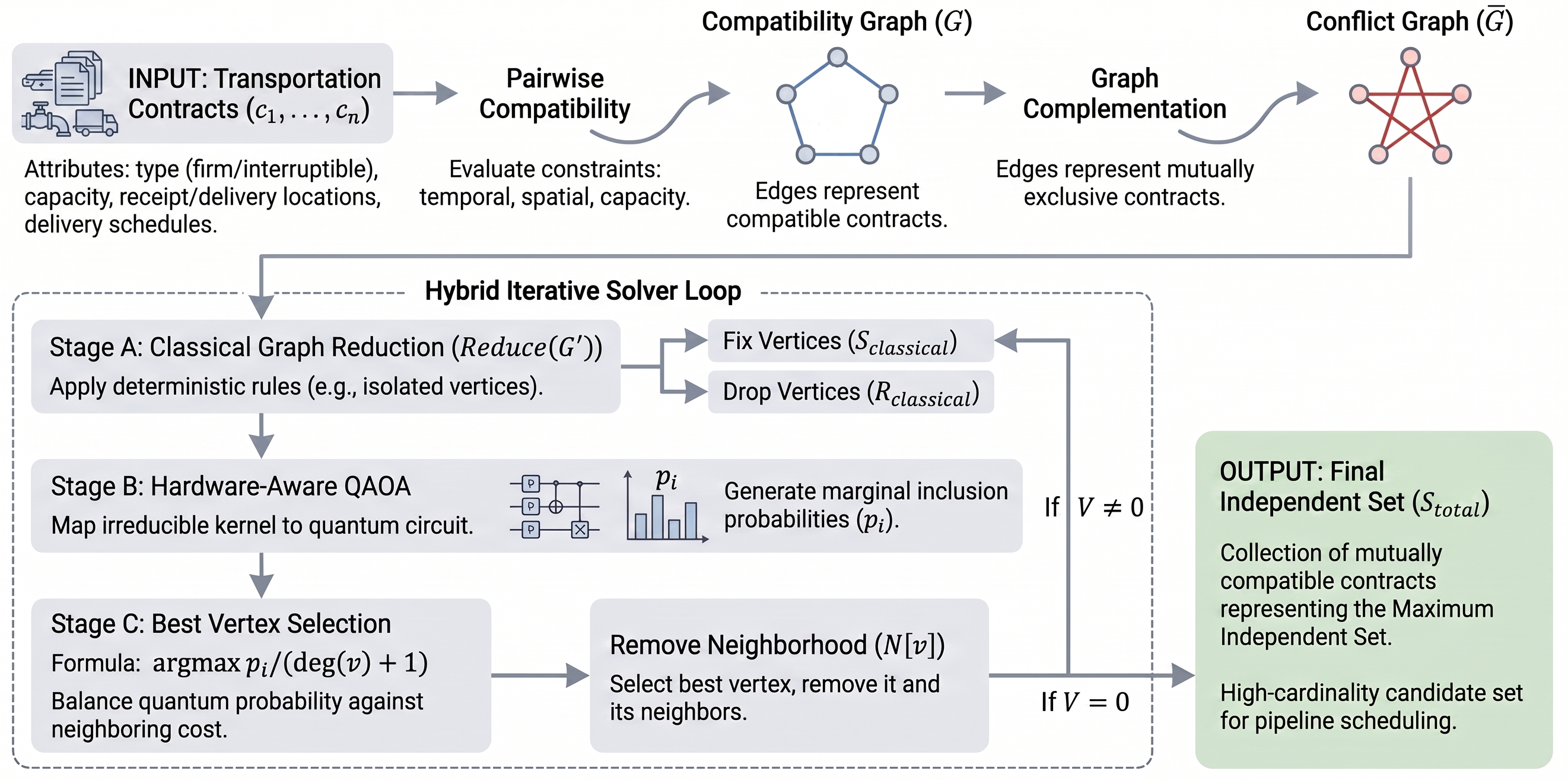}
\caption{Overview of the proposed hybrid quantum--classical workflow.
Transportation contracts form a compatibility graph $G$, whose complement
$\bar{G}$ is solved as a Maximum Independent Set (MIS) instance. The hybrid
solver then alternates between deterministic classical reduction, which fixes
provably-optimal vertices, and hardware-aware QAOA, whose measured marginal
probabilities $p_i$ rank the remaining vertices; the highest-scoring vertex is
added to the independent set and its closed neighborhood $N[v]$ is removed.
The reduced graph re-enters the loop until no vertices remain, at which point
the accumulated set $S_{\mathrm{total}}$, i.e., the constructed set of
mutually compatible contracts, is returned.}
\label{fig:pipeline}
\end{figure*}

\subsection{Problem Formulation}
As defined in Section~\ref{sec:introduction}, the Maximum Independent Set (MIS) problem seeks a maximum-cardinality subset $S\subseteq V$ of an undirected graph $G=(V,E)$ such that no two vertices in~$S$ are adjacent.

We define binary decision variables $x_i\in\{0,1\}$ for each vertex
$i\in V$, where $x_i=1$ indicates that vertex~$i$ is included in the
independent set. The MIS problem can be formulated as a Binary Linear Program (BLP)
\begin{equation}
\begin{aligned}
  \max_{x \in \{0, 1\}^n}   \quad & \sum_{i\in V} x_i \\
  \text{subject to} \quad & x_i + x_j \leq 1,
                        \quad \forall(i,j)\in E, \\
                        & x_i \in \{0,1\},
                        \quad \forall i\in V.
\end{aligned}
\end{equation}

Because the variables are binary, the adjacency constraint
$x_i+x_j\leq 1$ is equivalent to the quadratic constraint $x_ix_j=0$
for every $(i,j)\in E$. This representation is convenient for
constructing a QUBO
formulation, in which violations of the independence constraints are
incorporated into the objective through quadratic penalty terms. The QUBO formulation reads
\begin{equation}
  \min_{x \in \{0, 1\}^n}
  \quad  
  -\sum_{i\in V} x_i
  + P \sum_{(i,j)\in E} x_ix_j,
\end{equation}
where $P>1$ is a penalty parameter that ensures any optimal solution of the QUBO corresponds to a maximum independent set, and where we switched to minimization to follow common conventions. The first term encourages the selection of as many vertices
as possible; the second term penalizes simultaneous selection of
adjacent vertices.  The condition $P>1$ is sufficient because each selected vertex lowers the objective by exactly $1$, so any edge violation incurs a penalty $P>1$ that outweighs the gain of the additional vertex; every optimal QUBO solution is therefore
 edge-free and hence a maximum independent set.

For execution on quantum hardware, the QUBO is mapped to an Ising
Hamiltonian via the substitution $x_i=(1-z_i)/2$, where
$z_i\in\{-1,+1\}$ are spin variables. Replacing each $z_i$ with $Z_i$, the
Pauli-$Z$ operator acting on qubit~$i$, yields the cost Hamiltonian
\begin{equation}
  H_C = \sum_{i\in V} h_i Z_i
      + \sum_{(i,j)\in E} J_{ij} Z_i Z_j
      + \text{const}.
\end{equation}

where the
coefficients are given by
\begin{align}
  h_i &= \frac{1}{2}
        - \frac{P}{4}\sum_{j:(i,j)\in E}1\\
      &= \frac{1}{2} - \frac{P}{4}\deg(i),
  \\
  J_{ij} &= \frac{P}{4},
  \quad\forall(i,j)\in E,
\end{align}
with the constant offset
equal to $-\tfrac{1}{2}|V|+\tfrac{P}{4}|E|$.
Any ground state
of~$H_C$ encodes an optimal MIS solution, and $H_C$ serves as the
cost operator in the quantum optimization.

\subsection{Graph Preprocessing and Reduction}

The hybrid solver begins with a classical preprocessing stage that reduces the problem size before quantum optimization. Deterministic graph reduction rules eliminate vertices whose membership in a maximum independent set (MIS) can be decided exactly, after which the remaining graph is relabeled and decomposed into connected components. While the reduction may eliminate some globally optimal solutions when multiple optima exist, it is guaranteed to preserve at least one maximum independent set.

Given an input graph $G=(V,E)$, reduction rules are applied iteratively until no further simplifications are possible. The vertex set is partitioned into three disjoint subsets: vertices provably belonging to a MIS, denoted $S_{\mathrm{classical}}$; vertices excluded from further consideration, denoted $R_{\mathrm{classical}}$; and the remaining undecided vertices, which induce the reduced graph $G'=(V',E')$. Let $\deg(v)$ denote the degree of vertex $v$ and $N(v)=\{u\in V\mid(u,v)\in E\}$ its neighborhood. The following reduction rules are applied.

A vertex containing a self-loop cannot belong to an independent set:
\begin{equation}
(v,v)\in E
\;\Longrightarrow\;
v\in R_{\mathrm{classical}}.
\label{eq:selfloop}
\end{equation}

An isolated vertex can always be included:
\begin{equation}
\deg(v)=0
\;\Longrightarrow\;
v\in S_{\mathrm{classical}}.
\label{eq:isolated}
\end{equation}

A degree-one vertex is always preferred over its unique neighbour:
\begin{equation}
\deg(v)=1,\;N(v)=\{u\}
\;\Longrightarrow\;
\begin{cases}
v\in S_{\mathrm{classical}},\\
u\in R_{\mathrm{classical}}.
\end{cases}
\label{eq:degree1}
\end{equation}

If the neighborhood of a vertex forms a clique, the vertex dominates all of its neighbours and can be selected:
\begin{equation}
\forall\,u,w\in N(v),\;(u,w)\in E
\;\Longrightarrow\;
\begin{cases}
v\in S_{\mathrm{classical}},\\
N(v)\subseteq R_{\mathrm{classical}}.
\end{cases}
\label{eq:clique}
\end{equation}

The overall reduction is expressed as
\[
(G',\,S_{\mathrm{classical}},\,R_{\mathrm{classical}})
=
\textsc{Reduce}(G).
\]
Vertices in $S_{\mathrm{classical}}$ are immediately added to the partial independent set, while only the reduced graph $G'$ is forwarded to the quantum optimization stage.

To simplify the subsequent QUBO formulation and establish a one-to-one correspondence between graph vertices and qubits, the vertices of $G'$ are relabeled with consecutive integers. Let
\[
\phi:V_{\mathrm{original}}
\rightarrow
\{0,\ldots,|V'|-1\}
\]
denote the relabeling map and $\phi^{-1}$ its inverse. The relabeling preserves graph connectivity and enables measurement outcomes to be mapped back to the original vertex labels after quantum execution.

Finally, the reduced graph is decomposed into its connected components,
\[
G'_i=(V'_i,E'_i),
\qquad
i=1,\ldots,m,
\]
which are processed independently in non-increasing order of size,
$|V'_1|\ge|V'_2|\ge\cdots\ge|V'_m|$. Since no edges exist between distinct components, each component defines an independent optimization problem. This decomposition reduces the number of qubits required per quantum execution and enables multiple components to be solved in parallel when quantum resources are available.

\subsection{Hardware-Aware Quantum Circuit Generation}

To reduce routing overhead and limit circuit depth during compilation, the reduced problem graph is approximately mapped onto a hardware-compatible interaction graph before constructing the quantum circuit. Figure~\ref{fig:embedding} illustrates the overall procedure.

\begin{figure*}[t]
\centering
\includegraphics[width=\textwidth]{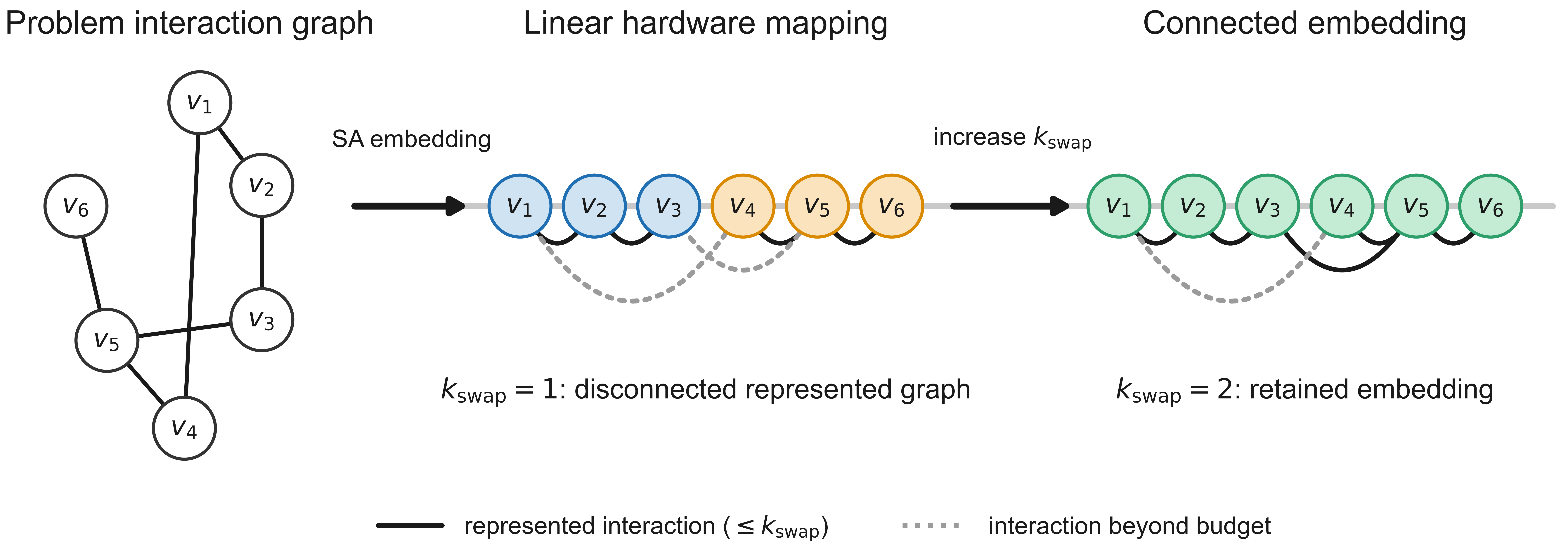}
\caption{Hardware-aware embedding of the problem interaction graph onto a
linear qubit topology. The problem interaction graph (left) is mapped by a
simulated-annealing heuristic onto a linear hardware-compatibility graph, in
which two qubits interact only if they lie within the SWAP budget
$k_{\mathrm{swap}}$. At a small budget (center) some interactions exceed the
budget and the represented graph is disconnected; increasing
$k_{\mathrm{swap}}$ (right) yields a single connected component, and the
smallest such budget is retained. Solid arcs denote interactions realizable
within the budget; dashed arcs denote interactions beyond it. Among the seeds
that achieve connectivity, the embedding with the smallest two-qubit depth is
selected.}
\label{fig:embedding}
\end{figure*}

A hardware compatibility graph containing the same number of vertices as the reduced problem graph is first constructed by modeling the target hardware as a linear qubit topology. In this graph, two qubits are considered adjacent if they can interact within a specified SWAP budget. A linear chain forms a subgraph of the heavy-hex coupling architecture used by current IBM superconducting quantum processors. Consequently, targeting this topology allows us to reach full connectivity in $|V'|-2$ layers of SWAP gates~\cite{de,crooks}.
Crucially, since noise deteriorates the performance of QAOA it is typically advantageous to generate the QAOA ansatz from a sub-set of the edges of the reduced graph~\cite{druagoi2026}.
A simulated-annealing heuristic is then used to map the QUBO interaction graph onto the hardware compatibility graph, maximizing the number of problem interactions that can be implemented within the SWAP budget while preserving connectivity of the resulting interaction graph. 

Rather than being fixed in advance, the SWAP budget $k_{\mathrm{swap}}$ is determined individually for each problem instance. Starting from a small value, $k_{\mathrm{swap}}$ is incrementally increased until the embedding produces a connected interaction graph, i.e., a single connected component without isolated qubits. The smallest budget satisfying this condition is selected. Since simulated annealing is inherently stochastic, different random seeds can generate different valid embeddings, resulting in varying two-qubit gate counts and circuit depths. After determining the minimum feasible SWAP budget, multiple random seeds are therefore evaluated, and the embedding producing the lowest two-qubit circuit depth is retained. This criterion is motivated by the fact that two-qubit gate count and depth are the primary contributors to noise accumulation on superconducting quantum hardware. The selected hardware-compatible interaction graph directly defines the Hamiltonian generating the QAOA circuit. Because the mapped interaction graph already conforms to the target hardware connectivity, transpilation requires little or no additional routing, producing shallower quantum circuits with reduced two-qubit gate overhead.
Related work has applied simulated annealing to reduce SWAP overhead in QAOA compilation. Montanez-Barrera~\textit{et al.}~\cite{montanezbarrera2025} use simulated annealing to optimize qubit ordering within SWAP-network and parity-twine-chain encodings, directly minimizing two-qubit gate count --similar in its use of annealing, but operating within a fixed circuit encoding rather than on an abstract hardware compatibility graph. Dupont~\textit{et al.}~\cite{dupont2024} embed $p{=}1$ QAOA subgraphs onto a linear qubit chain, sharing our line-topology target but using a SAT-solver-based isomorphism search instead of a stochastic heuristic, and without an instance-adaptive SWAP budget. U.S.\ Patent No.~11,537,770~\cite{patent11537770} describes a simulated-annealing procedure that generates candidate qubit-to-hardware subgraphs via depth-first search or swaps, accepted with a temperature-dependent probability -- algorithmically the closest match to our search mechanics, though it optimizes mapped-circuit fidelity rather than connectivity and edge coverage, and is not specific to a line topology or QAOA.

\subsection{QAOA Parameter Optimization}
\label{subsec:qaoa-training}
The Quantum Approximate Optimization Algorithm
(QAOA)~\cite{farhi2014quantum} of depth~$p$ is parameterized by
\[
\boldsymbol{\theta}
=
(\beta_1,\ldots,\beta_p,
 \gamma_1,\ldots,\gamma_p)
\]
where $\beta_i$ and $\gamma_i$ are the variational parameters associated with the mixer and cost layers, respectively.
These parameters define the variational quantum state $\ket{\psi(\boldsymbol{\theta})}$, which is obtained by applying the parameterized QAOA circuit to the initial uniform superposition state. The parameters are optimized to minimize the expectation value of the ideal QUBO cost Hamiltonian,
\begin{equation}
E(\boldsymbol{\theta})
=
\langle\psi(\boldsymbol{\theta})|H_C|\psi(\boldsymbol{\theta})\rangle,
\end{equation}
where $H_C$ denotes the original cost Hamiltonian derived directly from the QUBO formulation. Note that $H_C$ is distinct from the Hamiltonian implemented by the compiled quantum circuit. In our hardware-aware approach, the circuit realizes an approximate hardware-compatible Hamiltonian obtained after the interaction graph is modified to satisfy the connectivity constraints of the target device. Consequently, the optimized quantum state is prepared using the approximate compiled Hamiltonian, while its quality is evaluated with respect to the original objective Hamiltonian $H_C$.
Because the landscape is highly non-convex, a three-stage optimization strategy is employed, following operational guidance on parameter-setting methods for utility-scale QAOA~\cite{guo2026setting} and implemented using the open-source \texttt{qaoa\_training\_pipeline}~\cite{qaoatrainingpipeline}.
\paragraph{Stage~1: depth-one grid search.}
At $p=1$ the energy can be evaluated exactly and efficiently from the
graph structure alone, without full statevector
simulation~\cite{shaydulin2023evidence}. The parameter space is
scanned on a uniform $40\times40$ grid over
$\gamma\in[0,\pi/2]$ and $\beta\in[0,\pi]$; the lowest-energy pair
provides the initialization for Stage~2.
\paragraph{Stage~2: local refinement at $p=1$.}
Starting from the grid-search result, the derivative-free COBYLA
optimizer~\cite{powell1994direct} is run for up to 500 iterations
with initial trust-region radius $\rho=0.05$, yielding the optimized
depth-one parameters $(\gamma_1^*,\beta_1^*)$.
\paragraph{Stage~3: recursive depth extension.}
The depth-$k$ solution is extended to depth $k+1$ using the linear
interpolation scheme of Zhou~et~al.~\cite{zhou2020quantum}. Treating
the $k$ optimized parameters as samples at integer knot positions
$1,\ldots,k$, a new vector of $k+1$ values is formed by re-sampling at
$k+1$ evenly spaced points over $[1,k]$, with boundary values preserved
exactly and interior values obtained by linear interpolation. The
extended vector initializes COBYLA at depth~$k+1$. This step is
repeated until the target depth~$p$ is reached.
The target depth~$p$ is selected per instance in an incremental fashion.
The solver begins with a shallow circuit and increases~$p$ only when the
current depth fails to recover a satisfactory solution, exploiting the
fact that the recursive scheme reuses the depth-$k$ parameters to
warm-start depth~$k+1$. Easier kernels are therefore solved at low depth
($p=2$--$3$), while a larger depth (up to $p=5$) is reserved for the hardest
kernels, keeping the circuit as shallow as possible for each problem. This
accounts for the range of $p$ values reported in
Table~\ref{tab:circuit_characteristics}.
For $p\ge2$, the objective is evaluated using the Pauli propagation
method~\cite{fontana2023classical,shaydulin2023evidence,rudolph2023paulipropagation},
which evolves the observable~$H_C$ backwards through the circuit in the
Heisenberg picture and reads off the expectation value as the coefficient
of the identity string. Pauli strings with weight exceeding six or
absolute coefficient below $10^{-3}$ are discarded, keeping the
classical cost polynomial for the circuits considered here.

\subsection{Quantum-Guided Graph Reduction and Iterative Optimization}

The quantum component of the proposed hybrid solver is used as a
probabilistic oracle to provide vertex-selection information for iterative
graph reduction. The QAOA variational parameters are optimized entirely on a
classical computer using the Pauli propagation method described previously.
After obtaining the optimized parameter vector
\[
\boldsymbol{\theta}^*
=
(\beta_1^*,\ldots,\beta_p^*,
 \gamma_1^*,\ldots,\gamma_p^*),
\]
the parameters are bound to the circuit, which is subsequently executed on
IBM superconducting quantum hardware. Therefore, the quantum processor is
used exclusively for sampling candidate solutions, while the variational
optimization remains entirely classical.

For each reduced subproblem, the optimized circuit is executed for
$N_{\mathrm{shots}}$ measurements, generating a multiset of bitstrings
$\mathcal{M}=\{(b_i,n_i)\}$, where $b_i\in\{0,1\}^{n}$ and $n_i$ denotes the
number of occurrences of bitstring $b_i$. Each measured bitstring is
evaluated using the implemented Ising Hamiltonian
\[
H=\sum_i h_iZ_i+\sum_{(i,j)}J_{ij}Z_iZ_j,
\]
with corresponding energy
\[
E(b)=\sum_i h_is_i+\sum_{(i,j)}J_{ij}s_is_j,
\]
where $s_i=1$ for $b_i=0$ and $s_i=-1$ otherwise. The samples are ranked
according to increasing energy, with lower-energy states corresponding to
higher-quality candidate independent sets.

Rather than selecting only the single lowest-energy solution, the solver
exploits the statistical information contained in the low-energy measurement
distribution. Specifically, only a fraction $f$ of the lowest-energy
bitstrings is retained, corresponding to approximately
$N_{\mathcal{B}}\approx fN_{\mathrm{shots}}$ samples. The marginal inclusion
probability of vertex $i$ is estimated from this filtered sample set
$\mathcal{B}$ as
\[
p_i
=
\frac{1}{N_{\mathcal{B}}}
\sum_{b\in\mathcal{B},\,b_i=1}n_b,
\qquad
N_{\mathcal{B}}=\sum_{b\in\mathcal{B}}n_b .
\]
These marginal probabilities provide an estimate of the likelihood that a
vertex belongs to a high-quality independent set while reducing the impact
of high-energy noise-dominated measurements.

The estimated probabilities are then used to guide classical graph
reduction. Vertices with inclusion probability exceeding a confidence
threshold $\tau$ (set to $\tau=0.9$; Table~\ref{tab:settings}) are selected
directly. The remaining vertices are ranked using
\begin{equation}
\mathrm{score}(v)=\frac{p_v}{\deg(v)+1},
\label{eq:score}
\end{equation}
which balances the probability of selecting a vertex against the number of
vertices removed when it is included. Vertices are greedily selected in
descending score order while maintaining the independent-set constraint.

The score function admits an expected-gain interpretation. Let $I^{*}$
denote a target maximum independent set and define the indicator variable
\begin{equation}
X_v =
\begin{cases}
1, & v\in I^{*},\\
0, & \text{otherwise}.
\end{cases}
\end{equation}
Interpreting the estimated marginal probability as the probability that
vertex $v$ belongs to the target solution gives
\[
\mathbb{E}[X_v]=p_v .
\]
Selecting vertex $v$ removes its closed neighborhood
$N[v]$, containing
\[
|N[v]|=\deg(v)+1
\]
vertices from further consideration. Therefore, the expected contribution
per eliminated vertex is
\begin{equation}
\mathbb{E}[\Delta]
=
\frac{\mathbb{E}[X_v]}{|N[v]|}
=
\frac{p_v}{\deg(v)+1},
\end{equation}
which corresponds directly to Eq.~\eqref{eq:score}. Hence, the greedy
selection rule can be interpreted as maximizing the expected independent-set
contribution per unit of graph reduction.

The reliability of the estimated marginals depends on the number of retained
low-energy samples. Let $\hat{p}_v$ denote the empirical marginal estimated
from $N_{\mathcal{B}}$ samples. Since $\hat{p}_v$ is an empirical average of
bounded binary variables, Hoeffding's inequality gives
\begin{equation}
\Pr\!\left[|\hat{p}_v-p_v|\geq\epsilon\right]
\leq
2\exp(-2N_{\mathcal{B}}\epsilon^2).
\end{equation}
Applying a union bound over all $n=|V'|$ kernel vertices, all marginals are
estimated within an error $\epsilon$ with confidence $1-\delta$ provided
\begin{equation}
N_{\mathcal{B}}
\geq
\frac{1}{2\epsilon^2}
\ln\frac{2n}{\delta}.
\label{eq:samplecomplexity}
\end{equation}
Because the score function divides the marginal by
$\deg(v)+1\geq1$, the corresponding score estimates inherit the same
uniform accuracy bound. Thus, a sufficiently large measurement budget enables a reliable estimation of the $p_v$ produced by QAOA without explicit error mitigation \cite{Barron2024ProvableBounds}.

The measured marginals exhibit strong non-uniformity, indicating that the
quantum circuit provides informative vertex-selection probabilities rather
than random assignments. Figure~\ref{fig:marginals} shows the marginal
probabilities obtained during the first reduction round of QOBLIB instance
\texttt{es60fst02}. The probabilities range from values above $0.8$ to
almost zero, producing a clear separation between frequently and rarely
selected vertices. In contrast, uniform random sampling would assign
$p_i=0.5$ to all vertices. This separation provides the basis for the
quantum-guided reduction strategy.

\begin{figure}[t]
\centering
\includegraphics[width=\columnwidth]{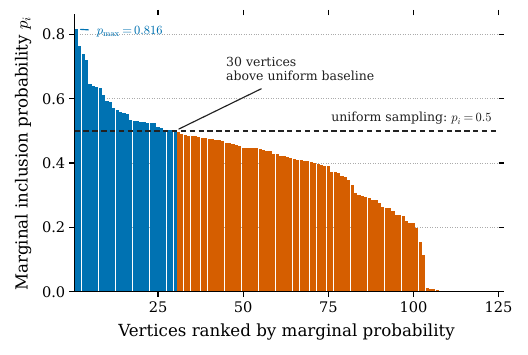}
\caption{QAOA marginal inclusion probabilities $p_i$ for the first reduction
round of \texttt{es60fst02} (124-vertex kernel), sorted in decreasing order.
The dashed horizontal line marks the value $p_i=0.5$ expected from uniform
random sampling; the $30$ vertices lying above this baseline (blue) carry the
selection signal, with the largest marginal reaching $p_{\max}=0.816$. The
non-uniform distribution of marginals provides the vertex-selection
information exploited by the proposed reduction strategy.}
\label{fig:marginals}
\end{figure}

The complete hybrid optimization procedure is iterative. Let
$G^{(0)}=(V^{(0)},E^{(0)})$ denote the input graph. At each iteration,
deterministic classical reduction rules are first applied to remove vertices
that can be safely selected or discarded. The remaining connected
components are then processed using the hardware-aware QAOA workflow and the
quantum-guided selection procedure described above. After selecting a set of
vertices $S^{(r)}$, the selected vertices and their neighbours are removed:
\[
G^{(r+1)}
=
G^{(r)}
\setminus
N[S^{(r)}],
\qquad
N[S]=S\cup\bigcup_{v\in S}N(v).
\]
The process is repeated until the graph becomes empty or the maximum
iteration limit is reached. The final independent set is obtained by
combining all vertices selected during deterministic reductions and
quantum-guided iterations. The complete procedure is summarized in
Algorithm~\ref{alg:hybrid}.

\begin{algorithm}[t]
\caption{Hybrid quantum-guided MIS solver}
\label{alg:hybrid}
\begin{algorithmic}[1]
\Require graph $G^{(0)}=(V,E)$; max iterations $R$; shots $N_{\mathrm{shots}}$;
         low-energy fraction $f$; confidence threshold $\tau$
\Ensure independent set $S_{\mathrm{total}}$
\State $S_{\mathrm{total}} \gets \emptyset$;\quad $r \gets 0$
\While{$V^{(r)} \neq \emptyset$ \textbf{and} $r < R$}
  \State $(G',S_{\mathrm{cl}},R_{\mathrm{cl}}) \gets \textsc{Reduce}(G^{(r)})$
  \State $S_{\mathrm{total}} \gets S_{\mathrm{total}} \cup S_{\mathrm{cl}}$
  \If{$V(G')=\emptyset$}
     \State \textbf{break}
  \EndIf
  \ForAll{connected components $G'_i$ of $G'$}
    \State build hardware-aware QAOA circuit for $G'_i$
    \State optimize $\boldsymbol{\theta}$ classically
    \State sample $N_{\mathrm{shots}}$ bitstrings on QPU
    \State retain low-energy fraction $f$ and estimate marginals $p_i$
    \State select vertices using $p_v/(\deg(v)+1)$ with threshold $\tau$
  \EndFor
  \State update graph by removing selected vertices and neighbors
  \State $r\gets r+1$
\EndWhile
\State \Return $S_{\mathrm{total}}$
\end{algorithmic}
\end{algorithm}

The algorithm terminates when either $V^{(r)}=\emptyset$ or the iteration
limit is reached. In practice, most instances are resolved before reaching
the maximum iteration count because deterministic reductions and
quantum-guided selections progressively decrease the graph size.

\subsection{Uniform Baseline for Quantum Sampling}

To quantify the contribution of the quantum-guided vertex selection mechanism, we introduce a uniform sampling variant of the proposed hybrid MIS solver. This experiment is designed as a controlled ablation study in which the quantum sampling stage is replaced by uniformly random vertex priorities while all other components of the algorithm remain unchanged. The objective is to isolate the information provided by the quantum stage and determine whether the quantum-generated signal offers meaningful guidance beyond uniform selection.

The uniform sampling baseline is therefore not intended as a comparison against competitive classical MIS heuristics. Replacing the quantum component with a specialized classical method, such as degree-based greedy selection or tabu search, would instead evaluate the relative performance of different node-selection strategies rather than the contribution of quantum information itself. The absolute performance of the proposed hybrid solver is evaluated separately against exact CPLEX solutions (cf.~Sections~
\ref{subsec:results-benchmark} and~\ref{subsec:results-gas}). Thus, this comparison should be interpreted as an ablation of the quantum-guided component rather than a benchmark against state-of-the-art classical MIS algorithms.

The uniform sampling baseline preserves the same iterative framework as the proposed hybrid method. At each iteration, the current graph is first processed using the identical classical reduction rules. Vertices that can be deterministically included in the independent set or safely removed are handled accordingly. If no vertices remain after reduction, the algorithm terminates. Otherwise, the quantum optimization stage is replaced by a randomized greedy selection procedure. Each remaining vertex is assigned an independent priority sampled from a uniform distribution. The vertex with the highest remaining priority is selected and added to the independent set if none of its neighbors has already been included. The selected vertex and all of its neighbors are then removed from the graph, and the reduction and randomized selection steps are repeated until the graph becomes empty. The final independent set is obtained by combining the vertices selected during the deterministic reduction stage with those selected through the randomized procedure.

Because the classical preprocessing, graph updates, and greedy selection mechanism are identical in both approaches, the only difference lies in the source of the vertex-selection information: quantum-generated samples in the proposed solver and uniformly random priorities in the ablation variant. Therefore, differences in solution quality can be attributed to the additional information provided by the quantum optimization stage. To account for the stochastic nature of the randomized baseline, it is executed using multiple random seeds, and the resulting distribution of independent-set sizes is compared with that obtained from the proposed hybrid solver.

The parameters governing the QUBO construction, quantum sampling, and
hybrid iteration are summarized in Table~\ref{tab:settings}. Unless stated otherwise, these values were held fixed across all reported experiments.

\subsection{Quantum Optimization Benchmarking Library (QOBLIB) Benchmark Instances}

The proposed hybrid MIS solver was evaluated on instances from QOBLIB~\cite{koch2026qoblib}. 
QOBLIB defines ten classically hard combinatorial optimization problem classes, including the Maximum Independent Set (MIS) problem, along with corresponding problem instances and reference solutions to enable fair, reproducible benchmarking. 
We selected 15 MIS instances from QOBLIB, varying in size and structure; see Table~\ref{tab:combined_results}.

\begin{table}[t]
\centering
\caption{Solver and quantum-execution settings used throughout the
experiments (unless stated otherwise).}
\label{tab:settings}
\begin{tabular}{lc}
\toprule
\textbf{Parameter} & \textbf{Value} \\
\midrule
QUBO penalty parameter $P$                     & $2$ \\
Shots per circuit $N_{\mathrm{shots}}$         & $20{,}000$ \\
Low-energy fraction retained for marginals $f$  & $0.5$ \\
Direct-selection confidence threshold $\tau$   & $0.9$ \\
Maximum hybrid iterations                      & $25$ \\
Random-baseline runs $N_{\mathrm{R}}$          & $10$--$40$ \\
$p=1$ grid-search resolution                   & $40\times40$ \\
COBYLA maximum iterations                      & $500$ \\
COBYLA initial trust-region radius $\rho$      & $0.05$ \\
Pauli-propagation weight cutoff                & $6$ \\
Pauli-propagation coefficient cutoff           & $10^{-3}$ \\
\bottomrule
\end{tabular}
\end{table}

\subsection{Pairwise Natural Gas Contract-Compatibility Abstraction}
\label{subsec:gas}

The proposed hybrid MIS solver is evaluated on a natural gas pipeline transportation contract-selection problem. The objective is to identify the largest set of mutually compatible transportation contracts that can be executed simultaneously. This problem is formulated as a Maximum Clique problem on a compatibility graph, where vertices represent contracts and edges connect pairwise compatible contracts. Since a maximum clique in a graph is equivalent to a maximum independent set in its complement, the proposed solver is applied to the complement graph.

Each transportation contract is characterized by its type (firm or interruptible), transportation capacity, transportation rate, receipt and delivery locations, pipeline segments, and delivery schedule. Synthetic benchmark instances are generated by randomly assigning one to four pipeline segments and delivery weeks to each contract while varying an overlap factor. Increasing the overlap factor increases competition for pipeline resources, producing denser compatibility graphs with more challenging conflict structures.

Two contracts are considered pairwise compatible if they can be executed simultaneously without violating the graph-generation rules, which account for temporal overlap, shared pipeline infrastructure, pairwise capacity-screening thresholds, contract priorities, and simulated flow conflicts. An edge is added between every compatible pair, yielding a compatibility graph whose complement defines the MIS instance solved by the proposed algorithm.

The resulting graph is a \emph{pairwise abstraction} of the underlying transportation problem. It does not explicitly enforce cumulative segment--period capacity constraints of the form
\[
\sum_{i\in S} a_{ist}\le C_{st},
\]
where $a_{ist}$ denotes the load contributed by contract $i$ on segment $s$ during period $t$, and $C_{st}$ is the corresponding segment capacity. Consequently, a clique in the compatibility graph may still violate aggregate capacity limits when three or more individually compatible contracts collectively exceed the available capacity. The industrial case study therefore evaluates the proposed algorithm on synthetic pairwise compatibility graphs, and the reported solutions should be interpreted as optimal with respect to the pairwise compatibility model rather than as fully validated pipeline schedules. Details of the benchmark generation process are provided in Appendix~\ref{app:benchmark_generation}.

The following proposition characterizes precisely when this pairwise abstraction is an exact representation of the underlying feasibility structure.

\begin{proposition}[Exactness of the pairwise compatibility graph]
\label{prop:exactness}
Let $\mathcal{F}$ denote the family of feasible contract subsets, and suppose feasibility is hereditary (downward closed), with every singleton contract feasible. Construct a compatibility graph $G_c=(V,E_c)$ in which $(i,j)\in E_c$ if and only if the pair $\{i,j\}$ is feasible. Then the cliques of $G_c$ coincide exactly with $\mathcal{F}$ if and only if every infeasible contract set contains an infeasible pair.
\end{proposition}

\begin{proof}
($\Leftarrow$) Assume every infeasible set contains an infeasible pair. If $S$ is a clique of $G_c$, then every pair in $S$ is feasible, so $S$ contains no infeasible pair and must therefore be feasible. Conversely, if $S\in\mathcal{F}$, hereditary feasibility implies that every pair of contracts in $S$ is feasible, so every pair forms an edge of $G_c$ and $S$ is a clique.

($\Rightarrow$) Assume the cliques of $G_c$ coincide with $\mathcal{F}$. If an infeasible set $S$ contained only feasible pairs, then every pair of vertices in $S$ would be adjacent, making $S$ a clique and hence feasible, a contradiction. Therefore every infeasible set must contain an infeasible pair.
\end{proof}

Proposition~\ref{prop:exactness} shows that the clique (equivalently, MIS) formulation is exact only when infeasibility arises solely from pairwise conflicts. It is no longer exact when higher-order resource constraints, such as cumulative segment--period capacity limits, render a set of pairwise compatible contracts jointly infeasible.

Rather than a shortcoming of the overall selection workflow, this motivates the intended role of the pairwise MIS solver as the first stage of a two-stage screening procedure. The purpose of the first stage is not to certify a fully feasible operating schedule, but to narrow an intractably large selection space---exponentially many contract subsets---down to a small family of mutually compatible, high-cardinality candidate sets that can then be checked against the cumulative constraints. In this second stage, each candidate independent set $S$ is verified against the cumulative segment--period capacity constraints $\sum_{i\in S} a_{ist}\le C_{st}$ for every pipeline segment $s$ and period $t$. This verification is inexpensive: it is a linear arithmetic check over the retained contracts rather than a combinatorial search. Candidate sets that satisfy all aggregate limits are directly feasible, whereas those that violate a segment--period limit can be repaired by removing the smallest-value conflicting contracts or by re-solving the MIS on the residual graph. This decomposition confines the computationally hard combinatorial search---where the quantum-guided solver is applied---to the pairwise-compatibility layer, while the aggregate-capacity constraints are enforced by the cheap downstream verification stage.

For this two-stage decomposition to be sound in the sense of discarding no truly feasible selection, the pairwise incompatibility rule must be a \emph{relaxation} of the underlying feasibility structure: a contract pair should be declared incompatible only when that pair on its own already violates a capacity or operational limit, which is a necessary condition for the joint feasibility of any set containing it. Under such a rule every feasible set remains pairwise compatible, so the compatibility graph over-approximates the feasible family and the first stage introduces no false negatives; the cumulative-capacity check in the second stage then removes the surviving false positives. We note that the synthetic generator used in this study instead applies a conservative pairwise capacity screen with a fractional threshold $\theta(\alpha)C_s<C_s$ (Appendix~\ref{app:benchmark_generation}), which is a proxy chosen to control graph density rather than a strict relaxation; a deployment aiming at the soundness guarantee above would set the pairwise rule accordingly. A full evaluation of the complete two-stage pipeline, including cumulative-capacity verification and hydraulic constraints on realistic network data, is left to future work.

\section{Results and Discussion}\label{sec:results}
The experimental evaluation of the proposed hybrid quantum--classical framework is organized into five parts. First (Section~\ref{subsec:results-benchmark}), we present benchmark results across fifteen diverse QOBLIB instances, demonstrating solution quality and comparing against the uniform sampling baseline. Second and third, we provide detailed case studies of the iterative hybrid optimization process on two representative benchmarks---C125-9 and \texttt{es60fst02}---illustrating the synergy between classical reduction and quantum-guided optimization and analyzing the statistical properties of the resulting solution distributions. Fourth (Section~\ref{subsec:runtime}), we report the runtime, resource usage, and run-to-run reproducibility of the solver on the QOBLIB instances. Fifth (Section~\ref{subsec:results-gas}), we evaluate performance on large-scale synthetic compatibility graphs representing natural gas contract-selection problems, validating practical applicability at industrially relevant scales. Characteristics of the transpiled QAOA circuits are summarized separately in Appendix~\ref{app:circuits}.

\begin{table*}[t]
\centering
\renewcommand{\arraystretch}{1.3}
\setlength{\tabcolsep}{5pt}
\caption{Comparison of uniform sampling and the proposed quantum--classical method across the benchmark instances. ``Best Approx. Ratio Run'' denotes the number of runs that achieved the method's best observed MIS size, reported as a fraction of the total number of runs. The approximation ratio is defined as the best MIS size obtained by the method divided by the CPLEX-optimal MIS size.}
\label{tab:combined_results}
\resizebox{\textwidth}{!}{%
\begin{tabular}{l r r r r c c c c}
\toprule
& & & & &
  \multicolumn{2}{c}{\textbf{Uniform Sampling}} &
  \multicolumn{2}{c}{\textbf{Quantum--Classical}} \\
\cmidrule(lr){6-7}\cmidrule(lr){8-9}
\textbf{Instance} & \textbf{Nodes} & \textbf{Edges} & \textbf{CPLEX} & \textbf{Runs} &
\makecell{\textbf{Best\ Approx.}\\\textbf{Ratio Run}} &
\makecell{\textbf{Approx.}\\\textbf{Ratio}} &
\makecell{\textbf{Best\ Approx.}\\\textbf{Ratio Run}} &
\makecell{\textbf{Approx.}\\\textbf{Ratio}} \\
\midrule
karate                    & 34  & 78   & 20 & 10 & 10/10 & 1.00 & 10/10 & 1.00 \\
farm                      & 17  & 39   & 10 & 10 & 6/10  & 1.00 & 10/10 & 1.00 \\
football                  & 35  & 118  & 16 & 10 & 2/10  & 1.00 & 10/10 & 1.00 \\
ibm32                     & 32  & 90   & 13 & 10 & 2/10  & 0.92 & 8/10  & 1.00 \\
es60fst01                 & 123 & 159  & 60 & 10 & 10/10 & 1.00 & 10/10 & 1.00 \\
chesapeake                & 39  & 170  & 17 & 10 & 2/10  & 1.00 & 2/10  & 1.00 \\
sloane\_1dc\_64           & 64  & 543  & 10 & 20 & 1/20  & 1.00 & 1/20  & 1.00 \\
insecta-ant-colony1-day38 & 56  & 1134 & 6  & 10 & 3/10  & 1.00 & 9/10  & 1.00 \\
sloane\_2dc\_128          & 128 & 5173 & 5  & 10 & 6/10  & 1.00 & 5/10  & 1.00 \\
es60fst03                 & 113 & 142  & 55 & 10 & 2/10  & 1.00 & 5/10  & 1.00 \\
es60fst04                 & 162 & 238  & 78 & 10 & 1/10  & 1.00 & 1/10  & 1.00 \\
sloane\_1dc\_128          & 128 & 1471 & 16 & 10 & 4/10  & 0.94 & 5/10  & 0.94 \\
sloane\_1zc\_128          & 128 & 1120 & 18 & 10 & 1/10  & 1.00 & 1/10  & 1.00 \\
es60fst02                 & 186 & 280  & 88 & 20 & 5/20  & 0.99 & 5/20  & 1.00 \\
C125-9                    & 125 & 787  & 34 & 40 & 1/40  & 0.94 & 2/40  & 1.00 \\
\bottomrule
\end{tabular}%
}
\end{table*}

\subsection{Benchmark Performance and Comparison with the Randomized Baseline}
\label{subsec:results-benchmark}

Table~\ref{tab:combined_results} evaluates 15 benchmark instances (17--186 vertices, 39--5{,}173 edges), comparing the exact CPLEX optimum against a uniform-sampling baseline and the proposed hybrid quantum--classical (Q--C) solver under identical run budgets.

The Q--C solver achieves an average best-solution approximation ratio of $\mathbf{0.996}$, recovering the CPLEX optimum on 14 of 15 instances ($93.3\%$). The sole exception is \texttt{sloane\_1dc\_128} (best MIS of 15 vs.\ optimum 16, ratio $0.938$). Optimal solutions are also recovered on larger, harder instances such as \texttt{es60fst02} and \texttt{C125-9} (up to 186 vertices), indicating that quantum guidance remains useful even when classical reduction leaves a substantial residual search space.

By contrast, uniform sampling fails to reach the optimum on four instances---\texttt{ibm32}, \texttt{sloane\_1dc\_128}, \texttt{es60fst02}, and \texttt{C125-9}---with ratios of $0.92$, $0.94$, $0.99$, and $0.94$, respectively. The Q--C solver closes this gap on three of the four (\texttt{ibm32}, \texttt{es60fst02}, \texttt{C125-9}), matching uniform sampling's $0.94$ only on \texttt{sloane\_1dc\_128}.

Search consistency also improves: across 200 total runs, the Q--C solver recovers its best observed solution in $42\%$ of runs versus $28\%$ for uniform sampling---a $50\%$ relative increase, with gains on 7 of 15 instances. The largest improvements are on \texttt{football} ($2/10 \to 10/10$), \texttt{farm} ($6/10 \to 10/10$), \texttt{ibm32} ($2/10 \to 8/10$), and \texttt{insecta-ant-colony1-day38} ($3/10 \to 9/10$); smaller gains appear for \texttt{es60fst03}, \texttt{sloane\_1dc\_128}, and \texttt{C125-9}. Frequencies are unchanged on six instances, and only \texttt{sloane\_2dc\_128} declines ($6/10 \to 5/10$).

The solution distributions show the clearest benefit. For \texttt{ibm32}, the Q--C solver reaches the optimum in 8/10 runs versus 2/10 for uniform sampling. For \texttt{es60fst02}, uniform sampling never reaches the optimum of 88, while Q--C does so in 5/20 runs. The largest shift is on \texttt{C125-9}: uniform sampling has a mean MIS of 28.05 and reaches size $\geq 31$ in only $5\%$ of runs, while Q--C raises the mean to 30.45, reaches size $\geq 31$ in $60\%$ of runs, and hits the optimum of 34 twice.

On easier instances---e.g., \texttt{karate} and \texttt{es60fst01}, where both methods reach the optimum every run---quantum guidance offers little added benefit, since uniform priorities already suffice.

Overall, classical reduction removes most of the search space, and the QPU supplies marginal probability information that biases the residual classical search toward promising regions. This raises the aggregate best-solution recovery frequency from $28\%$ to $42\%$ and the average approximation ratio from $0.986$ to $\mathbf{0.996}$, converting three previously suboptimal uniform-sampling cases into optimal ones. The benefit is most apparent when classical reduction leaves a sufficiently difficult residual problem.

\subsection{Detailed Case Study: C125-9 Benchmark}

\begin{table*}[t]
\centering
\caption{Round-wise evolution of the hybrid quantum--classical MIS solver on the C125-9 graph.
``After CR'' denotes the graph remaining after the classical reduction (CR) stage.
$\Delta\mathrm{MIS}_{\mathrm{C}}$ and $\Delta\mathrm{MIS}_{\mathrm{Q}}$ are the numbers of MIS vertices contributed by the classical reduction and quantum stages, respectively, in the current round.
``Neighbors Removed'' denotes the additional vertices eliminated because they are adjacent to the quantum-selected vertices.
``Net Reduction'' is the total number of vertices eliminated during the round, including classical reductions, selected MIS vertices, and their conflicting neighbors.
``Cumulative MIS'' records the number of vertices selected up to and including the current round.}
\label{tab:C125-9_results}
\small
\setlength{\tabcolsep}{3.5pt}
\begin{tabular}{c r r r r r r r r}
\toprule
\textbf{Round} &
\makecell{\textbf{Start}\\\textbf{Nodes}} &
\makecell{\textbf{After}\\\textbf{CR}} &
$\Delta\mathrm{MIS}_{\mathrm{C}}$ &
$\Delta\mathrm{MIS}_{\mathrm{Q}}$ &
\makecell{\textbf{Neighbors}\\\textbf{Removed}} &
\makecell{\textbf{Net}\\\textbf{Reduction}} &
\makecell{\textbf{Final}\\\textbf{Nodes}} &
\makecell{\textbf{Cumulative}\\\textbf{MIS}} \\
\midrule
1  & 125 & 125 & 0 & 1 & 6 & 7  & 118 & 1  \\
2  & 118 & 118 & 0 & 1 & 3 & 4  & 114 & 2  \\
3  & 114 & 114 & 0 & 1 & 5 & 6  & 108 & 3  \\
4  & 108 & 108 & 0 & 1 & 4 & 5  & 103 & 4  \\
5  & 103 & 103 & 0 & 1 & 5 & 6  & 97  & 5  \\
6  & 97  & 97  & 0 & 1 & 5 & 6  & 91  & 6  \\
7  & 91  & 91  & 0 & 1 & 5 & 6  & 85  & 7  \\
8  & 85  & 85  & 0 & 1 & 4 & 5  & 80  & 8  \\
9  & 80  & 80  & 0 & 1 & 6 & 7  & 73  & 9  \\
10 & 73  & 73  & 0 & 1 & 4 & 5  & 68  & 10 \\
11 & 68  & 68  & 0 & 1 & 3 & 4  & 64  & 11 \\
12 & 64  & 64  & 0 & 1 & 2 & 3  & 61  & 12 \\
13 & 61  & 61  & 0 & 1 & 4 & 5  & 56  & 13 \\
14 & 56  & 56  & 0 & 1 & 5 & 6  & 50  & 14 \\
15 & 50  & 50  & 0 & 1 & 2 & 3  & 47  & 15 \\
16 & 47  & 46  & 1 & 1 & 5 & 7  & 40  & 17 \\
17 & 40  & 37  & 3 & 1 & 6 & 10 & 30  & 21 \\
18 & 30  & 26  & 4 & 1 & 8 & 13 & 17  & 26 \\
19 & 17  & 16  & 1 & 1 & 4 & 6  & 11  & 28 \\
20 & 11  & 11  & 0 & 1 & 2 & 3  & 8   & 29 \\
21 & 8   & 3   & 5 & 0 & 3 & 8  & 0   & 34 \\
\bottomrule
\end{tabular}
\end{table*}

\subsubsection{Iterative Hybrid Optimization Process}

Table~\ref{tab:C125-9_results} presents the round-by-round execution of the proposed hybrid quantum--classical MIS solver on the C125-9 benchmark. Starting from the original graph containing 125 vertices, the algorithm iteratively applies classical reduction, executes the quantum MIS solver on the reduced graph, and removes the selected MIS vertex together with its neighbors. This process continues until the graph is completely eliminated after 21 rounds.

During the first fifteen rounds, the classical reduction stage is unable to identify any mandatory MIS vertices, resulting in $\Delta$MIS$_C=0$ throughout this phase. Consequently, the reduction in graph size is driven entirely by the quantum solver, which consistently contributes one vertex to the independent set in each iteration ($\Delta$MIS$_Q=1$). Depending on the degree of the selected vertex, between two and six neighboring vertices are removed in each round, reducing the graph from 125 vertices to 47 vertices while constructing an independent set of size 15.

As the graph becomes smaller, the classical reduction stage becomes increasingly effective. From Round 16 onward, classical reductions contribute additional MIS vertices with $\Delta$MIS$_C$ values of 1, 3, 4, 1, and 5, respectively. These reductions substantially accelerate graph simplification by identifying vertices that are guaranteed to belong to the maximum independent set without requiring further quantum optimization. The largest reduction occurs in Round 18, where four vertices are identified through classical reduction and one additional vertex is selected by the quantum solver. Together with the removal of eight neighboring vertices, the graph size decreases from 30 to 17 vertices in a single iteration. In the final round, the classical reduction identifies the remaining five MIS vertices, leaving a residual graph of three vertices that is eliminated without invoking the quantum solver ($\Delta$MIS$_Q=0$).

The hybrid solver constructs a maximum independent set of size 34 after 21 iterations while completely eliminating the graph. CPLEX verification confirms this is the optimal solution, demonstrating the complementary roles of the quantum and classical components: the quantum solver provides steady progress when classical reductions are initially ineffective, while the classical reduction stage becomes increasingly powerful as the graph shrinks.

\subsubsection{Statistical Analysis of Solution Quality}

\begin{table}[t]
\centering
\small
\setlength{\tabcolsep}{4pt}
\caption{Comparison of the solution quality achieved by the uniform-baseline solver and the optimized quantum solver over repeated executions on the C125-9 instance. In the first iteration, the optimized quantum circuit contained 35 controlled-\(Z\) (CZ) gates with a CZ-gate depth of 642.}
\label{tab:mis_summary}
\begin{tabular}{lcc}
\toprule
\textbf{Metric} &
\makecell{\textbf{Uniform}\\\textbf{Baseline}} &
\makecell{\textbf{Quantum}\\\textbf{-Classical}} \\
\midrule
Mean MIS Size              & 28.05 & 30.45 \\
Median MIS Size            & 28    & 31    \\
$P(\mathrm{MIS}\geq31)$ (\%) & 5.0   & 60.0  \\
$P(\mathrm{MIS}\geq32)$ (\%) & 2.5   & 35.0  \\
Best Observed MIS Size     & 32    & 34    \\
\midrule
Cliff's $\delta$            & \multicolumn{2}{c}{0.649} \\
\bottomrule
\end{tabular}
\end{table}

\begin{figure*}[t]
    \centering
    \begin{minipage}[b]{0.48\linewidth}
        \centering
        \includegraphics[width=\linewidth]{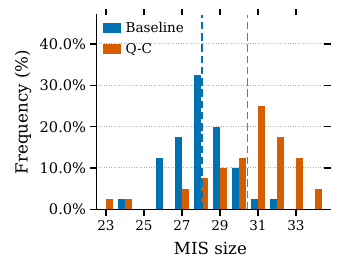}\\[2pt]
        (a)
    \end{minipage}
    \hfill
    \begin{minipage}[b]{0.48\linewidth}
        \centering
        \includegraphics[width=\linewidth]{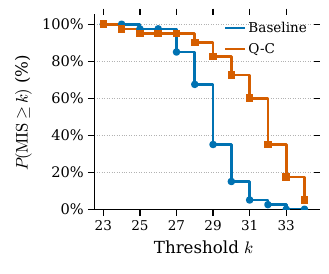}\\[2pt]
        (b)
    \end{minipage}
     \caption{Comparison of the uniform-random baseline (Baseline) and the hybrid quantum--classical (Q--C) solver on the distribution of maximum independent set (MIS) sizes for \texttt{C125-9}. (a) Histograms of the MIS sizes returned by the two methods over repeated runs; dashed vertical lines mark the respective mean MIS sizes. (b) Tail probability $P(\mathrm{MIS}\geq k)$---the fraction of runs attaining an MIS of size at least $k$---as a function of the threshold $k$.}
    \label{fig:mis-comparison}
\end{figure*}

\begin{table}[ht]
  \centering
  \small
  \setlength{\tabcolsep}{4pt}
  \caption{Tail probabilities $P(\mathrm{MIS} \geq k)$ for the uniform baseline and quantum solver on C125-9.}
  \label{tab:tail}
  \begin{tabular}{cccr}
    \hline
    \textbf{Threshold} &
    \makecell{\textbf{Uniform}\\\textbf{Baseline}} &
    \makecell{\textbf{Quantum}\\\textbf{-Classical}} &
    \textbf{Ratio (Q/R)} \\
    \textbf{$k$} & \textbf{$P(\mathrm{MIS}\geq k)$} &
    \textbf{$P(\mathrm{MIS}\geq k)$} & \\
    \hline
    31 & $5.0\%$  & $60.0\%$ & $12.0\times$ \\
    32 & $2.5\%$  & $35.0\%$ & $14.0\times$ \\
    33 & $0.0\%$  & $17.5\%$ & -- \\
    34 & $0.0\%$  & $5.0\%$  & -- \\
    \hline
  \end{tabular}
\end{table}

To evaluate the distributional advantage of quantum-guided optimization, we compare solution quality across repeated executions on C125-9. The random baseline ($N_{\mathrm{R}} = 40$ runs) and quantum algorithm ($N_{\mathrm{Q}} = 40$ runs, QAOA depth $p = 3$) produce markedly different distributions (Table~\ref{tab:mis_summary}, Figure~\ref{fig:mis-comparison}).

The quantum algorithm achieves substantially higher central tendency: a mean MIS size of $30.45$ versus $28.05$ for the random baseline (a difference of $2.40$ vertices), with median values of $31$ versus $28$. More significantly, the quantum approach exhibits pronounced advantages in the high-quality tail. At threshold $k = 31$ it is $12\times$ more likely to succeed ($60.0\%$ versus $5.0\%$), and this advantage widens to $14\times$ at $k = 32$ ($35.0\%$ versus $2.5\%$). At $k = 33$ and at the optimal size $k = 34$ the random baseline never succeeds, whereas the quantum algorithm still reaches these sizes in $17.5\%$ and $5.0\%$ of runs, respectively (Table~\ref{tab:tail}). Overall, the quantum algorithm places $60.0\%$ of its probability mass on solutions of size $31$ or larger, compared with only $5.0\%$ for the random baseline.

We quantify the separation between the two distributions with Cliff's
$\delta$, a nonparametric effect-size measure that requires no distributional
assumptions and is defined as
\begin{equation}
\delta = \frac{\bigl|\{(i,j): Q_i > R_j\}\bigr| - \bigl|\{(i,j): Q_i < R_j\}\bigr|}{N_{\mathrm{Q}}\cdot N_{\mathrm{R}}},
\end{equation}
where $Q_i$ and $R_j$ range over the $N_{\mathrm{Q}}$ quantum-classical and $N_{\mathrm{R}}$
uniform random MIS sizes. It equals the difference between the probability that a randomly
drawn quantum-classical solution exceeds a random baseline solution and the reverse
probability, so $\delta \in [-1, 1]$: values near $0$ indicate stochastically
identical distributions, while $|\delta| > 0.474$ conventionally denotes a large
effect~\cite{varghaDelaney2000}. Here $\delta = 0.649$ confirms a large effect size, indicating that the quantum-classical distribution stochastically dominates the random baseline to a substantial degree. In concrete terms, a randomly drawn quantum solution would be superior approximately 82.5\% of the time. These results provide strong evidence that quantum-guided optimization learns non-trivial search heuristics that bias sampling toward the high-quality region of the solution space, with the benefit most pronounced precisely where it matters most: at the frontier of best achievable solution quality.

\subsection{Detailed Case Study: \texttt{es60fst02} Benchmark}

\begin{table*}[tp]
\centering
\small
\setlength{\tabcolsep}{3.8pt}
\caption{Round-wise evolution of the proposed hybrid quantum--classical MIS solver on the \texttt{es60fst02} instance.
The \emph{Start Nodes} denote the number of vertices entering each reduction round, and \emph{After CR} denotes the residual graph after classical reduction (CR).
$\Delta\mathrm{MIS}_{\mathrm{C}}$ and $\Delta\mathrm{MIS}_{\mathrm{Q}}$ denote the MIS contributions from the classical and quantum stages, respectively.
\emph{Neighbors Removed} represents vertices eliminated because they conflict with vertices selected by the quantum stage.
The \emph{Final Nodes} are the vertices carried forward to the next round, while \emph{Net Reduction} denotes the total number of vertices eliminated during the round.
The final column reports the cumulative MIS size.}
\label{tab:graph_evolution}
\begin{tabular}{c r r r r r r r r}
\toprule
\textbf{Round} &
\makecell{\textbf{Start}\\\textbf{Nodes}} &
\makecell{\textbf{After}\\\textbf{CR}} &
$\Delta\mathrm{MIS}_{\mathrm{C}}$ &
$\Delta\mathrm{MIS}_{\mathrm{Q}}$ &
\makecell{\textbf{Neighbors}\\\textbf{Removed}} &
\makecell{\textbf{Net}\\\textbf{Reduction}} &
\makecell{\textbf{Final}\\\textbf{Nodes}} &
\makecell{\textbf{Cumulative}\\\textbf{MIS}} \\
\midrule
1 & 186 & 124 & 30 & 1 & 3 & 66 & 120 & 31 \\
2 & 120 & 120 & 0 & 1 & 2 & 3 & 117 & 32 \\
3 & 117 & 109 & 4 & 1 & 2 & 11 & 106 & 37 \\
4 & 106 & 102 & 2 & 1 & 2 & 7 & 99 & 40 \\
5 & 99 & 93 & 3 & 1 & 2 & 9 & 90 & 44 \\
6 & 90 & 77 & 7 & 1 & 2 & 16 & 74 & 52 \\
7 & 74 & 62 & 6 & 1 & 2 & 15 & 59 & 59 \\
8 & 59 & 57 & 1 & 1 & 3 & 6 & 53 & 61 \\
9 & 53 & 49 & 2 & 1 & 2 & 7 & 46 & 64 \\
10 & 46 & 0 & 24 & 0 & 0 & 46 & 0 & 88 \\
\bottomrule
\end{tabular}
\end{table*}

\begin{figure}[tbp]
  \centering
  \includegraphics[width=\columnwidth]{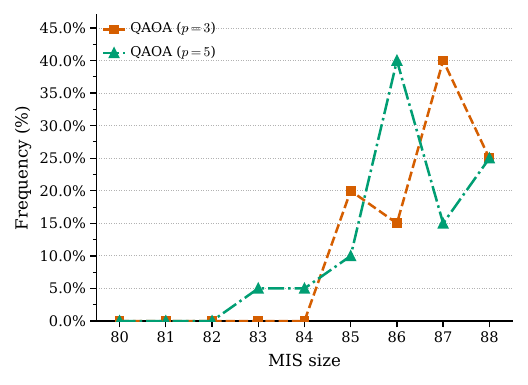}
  \caption{Distribution of independent set (MIS) sizes found on
  instance \texttt{es60fst02} by the QAOA-based solver at circuit depths
  $p = 3$ and $p = 5$ ($N = 20$ runs each). Frequencies are normalized within
  each configuration. The $p=3$ distribution peaks at $87$ ($40\%$ of runs)
  and the $p=5$ distribution at $86$ ($40\%$), and both reach the CPLEX
  optimum of $88$ in $25\%$ of runs; neither places any mass below $85$ ($p=3$)
  or $83$ ($p=5$). The uniform-random baseline over the same instance, which
  never exceeds size $87$, is tabulated in
  Table~\ref{tab:combined_results}.}
  \label{fig:mis_distribution}
\end{figure}

\subsubsection{Iterative Optimization Dynamics}

Table~\ref{tab:graph_evolution} presents the evolution of the hybrid solver on \texttt{es60fst02} (186 vertices). The algorithm achieves complete graph elimination after ten rounds, recovering the optimal MIS size of 88 (verified by CPLEX).

Classical reduction demonstrates exceptional effectiveness on this benchmark. The first round alone identifies 30 mandatory MIS vertices, reducing the graph from 186 to 124 vertices before quantum execution. Combined with quantum-guided selection of one additional vertex and neighbor removal, 66 vertices are eliminated in the first iteration. Classical reductions continue to dominate throughout: with the exception of Round 2, every iteration identifies additional mandatory vertices with $\Delta$MIS$_C$ values ranging from 1 to 24. The final iteration exemplifies this synergy---starting with 46 vertices, classical reduction identifies the remaining 24 MIS vertices without quantum invocation ($\Delta$MIS$_Q=0$), completing the solution.

This case study demonstrates the complementary strengths of the hybrid approach: classical reduction aggressively simplifies the graph and identifies a large fraction of the optimal solution (79 of 88 vertices), while quantum optimization resolves the remaining combinatorial ambiguity in only nine quantum executions.

\subsubsection{Impact of QAOA Depth on Solution Quality}

Figure~\ref{fig:mis_distribution} compares the distribution of maximum
independent set (MIS) sizes obtained on instance \texttt{es60fst02} by the
QAOA-based solver at circuit depths $p = 3$ and $p = 5$ ($N = 20$ runs each);
the uniform random baseline over the same instance ($N = 20$ runs, tabulated
in Table~\ref{tab:combined_results}) serves as the reference. Each distribution
is normalized to a relative frequency within its own set of runs.

The baseline distribution is centered near an MIS size of $85$ (mean $84.75$,
median $85.5$), with its upper mass split evenly between sizes $86$ and $87$
($25\%$ of runs each) and a lower tail extending down to $81$. Critically, it
never reaches the optimum of $88$, topping out at $87$. This behavior is
characteristic of an unbiased sampler, which spreads its probability mass
broadly over feasible configurations and only occasionally reaches high-quality
solutions.

In contrast, both QAOA configurations shift the distribution markedly toward
larger MIS sizes. The $p = 3$ solver peaks at a size of $87$ ($40\%$ of runs)
and returns no solution smaller than $85$, whereas the $p = 5$ solver peaks at
$86$ ($40\%$ of runs) with its support bounded below by $83$. Crucially, both
depths reach the optimal size of $88$ in $25\%$ of runs, whereas the baseline
never does. The mean MIS size rises to $86.70$ for $p = 3$ and $86.30$ for
$p = 5$, compared with $84.75$ for the baseline. The QAOA solver therefore both
raises the \emph{typical} solution quality (a larger mean and a mode shifted to
$86$--$87$ with a truncated lower tail) and materially increases the chance of
hitting the optimum, indicating that it concentrates probability mass on
high-quality solutions rather than sampling uniformly across the feasible
space.

These distributional differences are statistically significant. A
one-sided Mann--Whitney $U$ test confirms that both QAOA configurations
stochastically dominate the uniform random baseline: $p=3$ yields
$U=311.5$, $P=1.0\times10^{-3}$, with a large effect size (Cliff's
$\delta=0.557$), and $p=5$ yields $U=282.5$, $P=1.2\times10^{-2}$
(Cliff's $\delta=0.412$, medium-to-large). Bootstrap $95\%$ confidence
intervals on the mean MIS size ($20{,}000$ resamples) are
$[86.25,87.15]$ for $p=3$ and $[85.70,86.85]$ for $p=5$, both lying
entirely above the uniform random-baseline mean of $84.75$ (whose own interval,
$[83.85,85.60]$, does not reach either QAOA interval).

Increasing the circuit depth from $p = 3$ to $p = 5$ does not
yield a monotonic improvement in solution quality on this instance: the two
depths achieve the same probability of reaching size $88$ ($25\%$), and the
$p = 3$ configuration places more mass at size $87$ ($40\%$ versus $15\%$),
giving it a marginally higher mean. A plausible explanation is the increased noise arising from the deeper circuit, combined with the additional variational parameters at $p = 5$, which enlarge and roughen the optimization landscape---making the classical outer loop more susceptible to sub-optimal local minima within the available training budget. This
apparent difference is not statistically significant: a two-sided
Mann--Whitney test between the $p=3$ and $p=5$ distributions gives
$U=231.5$, $P=0.39$ (Cliff's $\delta=0.158$, negligible-to-small), and the
bootstrap $95\%$ confidence intervals on $P(\mathrm{MIS}\geq 88)$ are wide and
identical for the two depths ($[10\%,45\%]$ each), reflecting the small sample
size of $N=20$ runs. We therefore refrain from drawing conclusions about the
relative merits of $p=3$ versus $p=5$ from these data. The qualitative
conclusion---that the QAOA solver clearly outperforms the uniform baseline
on \texttt{es60fst02}---is robust, as the separation between the QAOA and
baseline distributions greatly exceeds this sampling uncertainty.

\subsection{Runtime and Reproducibility}
\label{subsec:runtime}

Table~\ref{tab:runtime} reports the resource usage and run-to-run
consistency of the hybrid solver on the QOBLIB instances. Each instance was
executed for a fixed number of independent runs, and we record how often the
solver recovered its own best solution across those runs as a measure of
reproducibility. Circuits were sampled on IBM superconducting processors of
the Heron (r2/r3) and Nighthawk (r1) families. The total wall-clock time is dominated by classical processing, including graph reduction, estimation of the required number of SWAP layers, seed selection to avoid excessively deep circuits, and QAOA parameter optimization using Pauli propagation. In contrast, QPU sampling requires only tens to a few hundred seconds per instance. Runtime grows with both the size of the irreducible kernel forwarded to
the QPU and the QAOA depth~$p$, which together govern the cost of the
classical Pauli-propagation parameter optimization; the largest classical
costs are accordingly incurred by the larger kernels (\texttt{sloane\_1zc\_128},
\texttt{sloane\_1dc\_128}, \texttt{C125-9}) and by the deeper $p=5$ circuit
(\texttt{sloane\_1dc\_64}). The hybrid solver reaches the certified optimum on fourteen of the fifteen instances; the sole exception is
\texttt{sloane\_1dc\_128}, whose best solution falls one vertex short of the
optimum.

Run-to-run reproducibility varies systematically with instance hardness. Small
or heavily reducible instances recover their best solution in every run (e.g.,
\texttt{karate}, \texttt{farm}, and \texttt{es60fst01}, each $10/10$), whereas the
largest and densest kernels are markedly less consistent: the solver attains its
own best reported value in only $1$ of $10$ runs for \texttt{es60fst04}, $1$ of
$10$ for \texttt{sloane\_1zc\_128}, and $2$ of $40$ for \texttt{C125-9}. This
variability reflects the stochastic nature of QAOA sampling on hard kernels and
is precisely why multiple independent runs are used: the best-of-$N$ values
reported in Table~\ref{tab:runtime} are obtained by exploiting this
run-to-run diversity, and the number of runs was increased for the harder
instances accordingly. The sole instance that never reaches the optimum is
\texttt{sloane\_1dc\_128}, which converges to a sub-optimal independent set in
all $10$ runs, recovering its own best value in $5$ of them, as discussed in
Section~\ref{subsec:results-benchmark}.

\begin{table*}[t]
\centering
\small
\setlength{\tabcolsep}{4pt}
\caption{Runtime and reproducibility summary for the QOBLIB benchmarks. The number of binary variables equals the number of graph
vertices. ``Times best found'' is the number of runs (out of the total) in
which the solver attained its own best reported solution. ``IBM QPU'' is the
superconducting backend used for sampling. Runtimes are wall-clock seconds}
\label{tab:runtime}
\begin{tabular}{l c r c c c c l r r}
\toprule
\textbf{Instance} &
\textbf{\makecell{Bin.\\vars}} &
\textbf{\makecell{Non-zero\\coef.}} &
\textbf{\makecell{Optimal\\MIS}} &
\textbf{\makecell{Hybrid\\MIS}} &
\textbf{Runs} &
\textbf{\makecell{Times\\best found}} &
\textbf{IBM QPU} &
\textbf{\makecell{CPU\\(s)}} &
\textbf{\makecell{QPU\\(s)}} \\
\midrule
karate                     & 34  & 112  & 20 & 20 & 10& 10 & Heron r3     & 254    & 7   \\
farm                       & 17  & 56   & 10 & 10 & 10& 10 & Heron r2     & 138    & 7   \\
football                   & 35  & 153  & 16 & 16 & 10& 10 & Heron r3     & 463    & 7   \\
ibm32                      & 32  & 122& 13 & 13 & 10& 8  & Nighthawk r1 & 8799   & 246 \\
es60fst01                  & 123 & 282  & 60 & 60 & 10& 10 & Heron r3     & 4275   & 28  \\
chesapeake                 & 39  & 209  & 17 & 17 & 10& 6  & Nighthawk r1 & 1643   & 164 \\
sloane\_1dc\_64            & 64  & 607  & 10 & 10 & 10& 1  & Heron r3     & 41498  & 24  \\
insecta-ant-colony1-day38  & 56  & 1190 & 6  & 6  & 10& 9  & Heron r3     & 12654  & 246 \\
sloane\_2dc\_128           & 128 & 5301 & 5  & 5  & 10& 5  & Nighthawk r1 & 3673   & 82  \\
es60fst03                  & 113 & 255  & 55 & 55 & 10& 5  & Heron r3     & 12372  & 64  \\
es60fst04                  & 162 & 400  & 78 & 78 & 10& 1  & Heron r3     & 21444  & 64  \\
sloane\_1dc\_128           & 128 & 1599 & 16 & 15 & 10& 5  & Nighthawk r1 & 54901  & 656 \\
sloane\_1zc\_128           & 128 & 1248& 18 & 18 & 10& 1  & Nighthawk r1 & 58745  & 820 \\
es60fst02                  & 186 & 466  & 88 & 88 & 20& 5  & Heron r2     & 23288  & 78  \\
C125-9                     & 125 & 912  & 34 & 34 & 40& 2  & Heron r2     & 46574  & 160 \\
\bottomrule
\end{tabular}
\end{table*}

\subsection{Performance on Large-Scale Pairwise Contract-Compatibility Graphs}
\label{subsec:results-gas}

\begin{table*}[t]
\centering
\small
\setlength{\tabcolsep}{6pt}
\renewcommand{\arraystretch}{1.15}
\caption{Performance of the proposed hybrid quantum--classical algorithm on synthetic compatibility graphs representing natural gas transportation contracts. The reduction percentage indicates the fraction of vertices eliminated during the classical preprocessing stage before quantum optimization.`No. of qubits' indicates the number of qubits in the quantum circuit during the first round of iterative reduction.}
\label{tab:performance}

\begin{tabular}{cccccccc}
\toprule
&
\multicolumn{3}{c}{\textbf{Graph Statistics}}
&
\multicolumn{4}{c}{\textbf{Solution Quality}}\\
\cmidrule(lr){2-4}
\cmidrule(lr){5-8}

\textbf{Graph}
&
\textbf{Contracts}
&
\makecell{\textbf{No. of}\\\textbf{qubits}}
&
\makecell{\textbf{Reduction}\\\textbf{(\%)}}
&
$\mathbf{MIS}_{\mathbf{hybrid}}$
&
$\mathbf{MIS}_{\mathbf{CPLEX}}$
&
\makecell{\textbf{Approx.}\\\textbf{Ratio}}
&
\makecell{\textbf{Baseline Solver}\\
$\mathbf{P(MIS \geq MIS_{CPLEX})}$ (\%)}\\

\midrule
1 & 280 &  23 & 91.8 &  58 &  58 & 1.000 & 35.5 \\
2 & 250 &  40 & 84.0 &  56 &  56 & 1.000 & 25.9 \\
3 & 500 &  63 & 87.4 &  79 &  79 & 1.000 & 7.4 \\
4 & 450 &  84 & 81.3 &  81 &  81 & 1.000 & 11.8 \\
5 & 800 & 108 & 86.5 &  98 & 100 & 0.980 & 1.1 \\
6 & 900 & 124 & 86.2 & 102 & 104 & 0.981 & 0.7 \\
\midrule
\textbf{Average}
& -- & -- & -- & -- & -- & \textbf{0.989} & -- \\
\bottomrule
\end{tabular}
\end{table*}

Table~\ref{tab:performance} summarizes the performance of the proposed
hybrid quantum--classical solver on six synthetic pairwise
contract-compatibility graphs containing between 250 and 900 contracts.
The benchmark instances are designed to emulate graph sizes, densities,
temporal overlaps, and shared-infrastructure patterns representative of
transportation contract-selection problems. The reported CPLEX and
hybrid solutions correspond to the generated MIS instances derived from
the pairwise compatibility model. Since compatibility is determined
solely through pairwise screening rules, the resulting independent sets
should be interpreted as optimal or near-optimal solutions to the
pairwise graph model rather than as fully validated pipeline operating
schedules under complete capacity constraints. Compared with the
benchmark instances in Section~\ref{subsec:results-benchmark}, these
industrial graphs are substantially larger and denser, reflecting the
complexity of realistic contract compatibility networks.

The classical reduction stage proves highly effective, eliminating an
average of 86.2\% of the vertices (Table~\ref{tab:performance}).
Consequently, the residual quantum kernels remain small---at most
124 vertices---allowing even the largest 900-contract instance to be
mapped onto the available quantum hardware.

Despite the increased problem complexity, the proposed hybrid solver
achieves an average approximation ratio of $\mathbf{0.989}$ and exactly
matches the CPLEX optimum on four of the six instances. For the two
largest graphs, the solver attains approximation ratios of 0.980 and
0.981, differing from the optimal solution by only two contracts in
each case, even though the corresponding quantum kernels contain more
than 100 vertices.

The uniform random baseline performs competitively only on the smallest
instances, and its success probability declines rapidly as the graph size
increases, falling below 1.1\% for the two largest (Table~\ref{tab:performance}).
In contrast, the proposed hybrid solver consistently maintains
approximation ratios exceeding 0.98, demonstrating that quantum-guided
vertex selection provides a substantial advantage over uninformed random
sampling for large-scale MIS instances.

Overall, these results demonstrate that the proposed framework scales
effectively to large industrial graph instances. The classical
preprocessing stage dramatically reduces the problem size while
preserving solution quality, leaving only compact residual kernels for
quantum optimization. The remaining optimality gap is therefore largely
attributable to the difficulty of the reduced quantum subproblems,
suggesting that continued advances in quantum hardware, circuit
optimization, and variational algorithms can further improve the quality
of hybrid quantum--classical MIS solvers.

It is worth clarifying the intended role of the quantum component in this
context. On the instance sizes considered here, exact solvers such as
CPLEX return provably optimal solutions and serve as our ground-truth
reference; we do not claim to surpass them in this regime. The
motivation is instead one of scalability: exact branch-and-cut solvers
scale exponentially in the worst case and become intractable on
sufficiently large and dense conflict graphs---for contract-selection
problems this regime is reached as the number of contracts grows into the
thousands. In that regime the hybrid framework remains applicable,
because classical reduction confines the quantum workload to a compact
kernel whose size is bounded by the available qubit count rather than by
the full problem size. Consequently, as quantum hardware continues to
grow, the hybrid solver is positioned to address instances that lie
beyond the practical reach of exact classical methods.

\section{Conclusion}\label{sec:conclusion}

This work presents a hybrid classical--quantum framework for the Maximum Independent Set problem that combines deterministic graph reduction with quantum-guided optimization. Classical exact methods have exponential worst-case complexity, while heuristics lack optimality guarantees. Purely quantum approaches, meanwhile, remain constrained by limited qubit connectivity, circuit depth, and noise on current hardware.

The framework comprises four components. First, deterministic graph reduction rules identify and fix vertices whose membership in an MIS can be established exactly, reducing the problem before quantum optimization. Second, hardware-aware circuit generation maps the reduced graph to the processor topology to limit routing overhead and circuit depth. Third, hierarchical QAOA parameter optimization recursively expands circuit depth using previously optimized parameters, reducing optimization cost and improving convergence. Fourth, quantum-guided node selection uses the full measurement distribution, rather than individual lowest-energy bitstrings, to extract statistical information for probabilistic graph reduction.

Experiments on fifteen benchmark graphs yield an average approximation ratio of 0.996, with optimal solutions obtained for fourteen instances, including graphs with up to 186 vertices. On C125-9, the quantum-guided approach significantly outperforms the uniform baseline, achieving a large effect size (Cliff's $\delta=0.649$) and concentrating markedly more probability mass in the high-quality tail (for instance, $12\times$ more likely to reach size $31$ or larger, and recovering the optimum where the baseline never does). The complementary roles of the classical and quantum stages are illustrated by \texttt{es60fst02}, where classical reduction eliminates up to 62 vertices in a single iteration and quantum optimization resolves the remaining combinatorial ambiguity.

The framework was further evaluated on synthetic pairwise natural-gas contract-compatibility graphs with up to 900 contracts. These experiments validate the MIS abstraction rather than a complete pipeline scheduling model incorporating cumulative segment-period capacity and hydraulic constraints. Accordingly, the MIS stage serves as the first step of a two-stage screening procedure: it identifies small families of mutually pairwise-compatible, high-cardinality contract sets that can subsequently be verified against aggregate pipeline constraints. The hybrid solver matches the CPLEX optimum on four of six instances and remains within two contracts on the others, while consistently outperforming an uninformed random frozen-node baseline. Classical preprocessing reduces the problem size by factors of 4--7, enabling hardware-sized kernels to be processed on current quantum devices.

Overall, the results show that iterative classical--quantum feedback can effectively address larger MIS instances than would be feasible using quantum optimization alone. By restricting quantum computation to irreducible kernels and using measurement outcomes to guide subsequent reductions, the framework progressively shrinks the problem while constructing a valid independent set. Its hardware-aware architecture can naturally benefit from future improvements in qubit count, connectivity, and coherence, providing a practical pathway for applying near-term quantum computing to large combinatorial optimization problems.

We close by restating the intended reading of these results. We do not present a faster MIS solver: at hardware-feasible sizes, classical exact methods such as branch-and-reduce and KaMIS/ReduMIS remain both quicker and provably optimal. The contribution is the method---hardware-aware embedding, the hierarchical QAOA schedule, and quantum-guided reduction---which is not specific to MIS and is meant to transfer to near-term quantum optimization more broadly. MIS is a deliberately chosen testbed: its known optima and strong classical reductions make the quantum contribution measurable rather than merely plausible, and they underwrite a hardware demonstration with ground truth at up to 125 qubits. Because the quantum workload here scales with the size of the irreducible kernel---that is, with available qubit count---rather than with the full problem size, the same architecture is positioned to follow improvements in hardware into the dense regimes where exact classical methods lose their footing. Establishing where that crossover lies, and on which graph families it arrives first, is the natural next question.

\section{Future Work}

Several directions can extend the framework. First, advanced error-mitigation methods, including zero-noise extrapolation and probabilistic error cancellation, could improve solution quality on noisy hardware. Second, problem-inspired mixers and adaptive variational ans\"{a}tze  could improve performance for specific graph classes. Third, machine-learning-assisted parameter transfer could provide more effective QAOA warm starts for unseen instances and reduce optimization overhead. Fourth, extending the approach to related problems such as maximum clique, graph coloring, and vertex cover would assess its broader applicability. Finally, theoretical analysis of approximation guarantees and convergence properties could clarify the conditions under which the hybrid strategy provides advantages over classical and quantum baselines and establish stronger performance guarantees.

\begin{acknowledgments}
We thank Stefan W\"orner for insightful discussions that shaped our embedding of the problem graph and its transpilation. We thank Daniel Egger for his help implementing the code used to optimize the quantum-circuit parameters. We further thank Stefan W\"orner, Daniel Egger, and Christa Zoufal for valuable discussions and for their careful review of the manuscript.
\end{acknowledgments}

\appendix
\section{Generation of Natural Gas Transportation Compatibility Graphs}
\label{app:benchmark_generation}

Synthetic compatibility graphs were generated to model the contract selection problem encountered in natural gas transportation systems. The objective is to create realistic benchmark instances whose graph structure resembles scheduling and resource allocation problems arising in pipeline transportation while allowing control over graph size and density.

Each graph vertex corresponds to a transportation contract. A contract is characterized by the following attributes:

\begin{itemize}
\item \textbf{Contract type:} Firm (60\%) or Interruptible (40\%). A \emph{Firm} contract provides a guaranteed transportation capacity under the agreed terms, whereas an \emph{Interruptible} contract provides transportation capacity that may be curtailed when pipeline capacity is constrained.
\item \textbf{Transportation capacity:}
\begin{itemize}
\item Firm: 50,000--300,000 MMBtu/day,
\item Interruptible: 20,000--150,000 MMBtu/day.
\end{itemize}
Here, \textbf{MMBtu/day} denotes \emph{million British thermal units per day}, a standard unit for expressing the daily quantity of natural gas transported through a pipeline.
\item \textbf{Transportation rate:}
\begin{itemize}
\item Firm: \$0.05-\$0.15 per MMBtu,
\item Interruptible: \$0.10-\$0.25 per MMBtu.
\end{itemize}
\item \textbf{Receipt and delivery locations} selected from representative natural gas hubs.
\item \textbf{Pipeline segments} traversed by the contract.
\item \textbf{Delivery schedule}, specified as a subset of weeks within a 52-week planning horizon.
\end{itemize}

The transportation network consists of multiple pipeline segments, each assigned a maximum daily throughput between 100,000 and 1,000,000 MMBtu/day. Contract routes are generated by assigning one to four pipeline segments, while delivery schedules are sampled from durations of 1, 4, 8, 12, 26, or 52 weeks. Seasonal demand is incorporated by partitioning the planning horizon into winter (Weeks 1--10 and 48--52), summer (Weeks 24--36), and shoulder seasons (Weeks 11--23 and 37--47). An overlap factor $\alpha\in[0,1]$ controls the concentration of contracts on common pipeline segments and peak-demand periods, thereby regulating the degree of competition among contracts.

A compatibility graph $G=(V,E)$ is constructed by representing each contract as a vertex. Two contracts are connected by an edge if they can be executed simultaneously without violating operational constraints. Compatibility is determined using the following rules:

\begin{enumerate}
\item \textbf{Temporal compatibility:} Contracts with non-overlapping delivery schedules are always compatible.
\item \textbf{Spatial compatibility:} Contracts using disjoint pipeline segments are always compatible.
\item \textbf{Capacity constraints:} For every shared pipeline segment,
\[
\kappa_i+\kappa_j \leq \theta(\alpha)C_s,
\]
where $\kappa_i$ and $\kappa_j$ denote the transportation capacities of the two contracts, $C_s$ is the capacity of pipeline segment $s$, and
\[
\theta(\alpha)=0.8-0.2\alpha.
\]
If both contracts are firm, a stricter limit of $(\theta(\alpha)-0.1)C_s$ is enforced.
\item \textbf{Operational compatibility:} Contracts sharing pipeline segments are additionally screened for flow conflicts using a probabilistic model that captures operational interference within the network.
\end{enumerate}

An edge is added to the compatibility graph only when all compatibility conditions are satisfied. The resulting graph models mutually compatible transportation contracts. Since the contract selection problem is formulated as a Maximum Clique problem, the complement of the compatibility graph is constructed and solved as a Maximum Independent Set problem using the proposed hybrid quantum--classical algorithm.

\begin{table*}[t]
\centering
\small
\setlength{\tabcolsep}{6pt}
\renewcommand{\arraystretch}{1.15}
\caption{Characteristics of the transpiled QAOA circuits for the benchmark instances. The reported values correspond to the circuit executed on the quantum backend in the first round of the hybrid reduce--QAOA--select loop, after transpilation.}
\label{tab:circuit_characteristics}

\begin{tabular}{lcccccccc}
\toprule
\multirow{2}{*}{\textbf{Benchmark}} &
\multicolumn{2}{c}{\textbf{Graph Statistics}} &
\multicolumn{2}{c}{\textbf{QAOA Parameters}} &
\multicolumn{2}{c}{\textbf{Circuit Size}} &
\multicolumn{2}{c}{\textbf{Two-Qubit Operations}} \\
\cmidrule(lr){2-3}
\cmidrule(lr){4-5}
\cmidrule(lr){6-7}
\cmidrule(lr){8-9}
&
\textbf{Nodes} &
\makecell{\textbf{No. of}\\\textbf{Qubits}} &
\textbf{$p$} &
\textbf{$k_{\mathrm{swap}}$} &
\makecell{\textbf{Total}\\\textbf{Gates}} &
\makecell{\textbf{Total}\\\textbf{Depth}} &
\makecell{\textbf{2Q}\\\textbf{Depth}} &
\makecell{\textbf{CZ}\\\textbf{Gates}} \\
\midrule
karate                    &  34 &   4 & 2 &  2 &   278 &  59 &   8 &   12 \\
farm                      &  17 &  10 & 2 &  4 &   522 & 121 &  20 &   40 \\
football                  &  35 &  15 & 2 &  2 &   647 &  59 &   8 &   56 \\
ibm32                     &  32 &  32 & 3 &  2 &  1713 & 122 &  18 &  186 \\
es60fst01                 & 123 &  34 & 2 &  4 &  1385 & 91  &  14 &  134 \\
chesapeake                &  39 &  39 & 3 &  4 &  2533 & 259 &  45 &  272 \\
sloane\_1dc\_64           &  64 &  50 & 5 &  2 &  4246 & 200 &  30 &  490 \\
insecta-ant-colony1-day38 &  56 &  54 & 3 &  2 &  2817 & 122 &  18 &  318 \\
sloane\_2dc\_128          & 128 &  70 & 2 &  2 &  2516 &  83 &  12 &  276 \\
es60fst03                 & 113 &  74 & 3 &  8 &  4408 & 256 &  46 &  483 \\
es60fst04                 & 162 &  85 & 3 &  8 &  4984 & 264 &  46 &  555 \\
sloane\_1dc\_128          & 128 & 112 & 3 &  4 &  6670 & 245 &  42 &  780 \\
sloane\_1zc\_128          & 128 & 112 & 3 &  4 &  6611 & 273 &  49 &  767 \\
es60fst02                 & 186 & 124 & 3 & 12 &  7028 & 261 &  45 &  772 \\
C125-9                    & 125 & 125 & 3 &  3 &  6170 & 199 &  35 &  642 \\
\bottomrule
\end{tabular}
\end{table*}

\section{Hardware-Aware Quantum Circuit Characteristics}
\label{app:circuits}

Table~\ref{tab:circuit_characteristics} summarizes the characteristics of the transpiled QAOA circuits executed on the IBM quantum backend. The benchmark graphs range from 17 to 186 vertices. Following the initial application of the classical MIS reduction rules, the resulting irreducible kernels contain between 4 and 125 vertices, corresponding directly to the number of qubits required by the quantum circuit. For many benchmarks, the preprocessing stage substantially reduced the quantum resource requirements. For example, the \textit{es60fst01} and \textit{es60fst02} instances were reduced from 123 and 186 vertices to kernels of only 34 and 124 vertices, respectively. In contrast, instances such as \textit{ibm32} and \textit{C125-9} admitted little or no reduction, requiring the quantum circuit to operate on the complete graph. These results demonstrate the effectiveness of the reduction framework in extending the range of problem sizes that can be executed on current quantum hardware.

The proposed hardware-aware circuit synthesis approach successfully generated executable QAOA circuits for all benchmark instances while maintaining moderate circuit depths. The transpiled circuits contain between 278 and 7\,028 gates, with total circuit depths ranging from 59 to 273. As expected, the circuit complexity increases with both the kernel size and the selected QAOA depth ($p$). For example, the \textit{sloane\_1dc\_64} instance, executed with $p=5$, requires 4\,246 gates and a two-qubit depth of only 30, whereas the larger \textit{es60fst02} instance ($p=3$, 124-vertex kernel) is the largest transpiled circuit, with 7\,028 gates, a two-qubit depth of 45, and 772 CZ gates.

Two-qubit operations are the dominant source of noise on superconducting quantum processors and therefore provide an important indicator of circuit quality. Across all benchmarks, the transpiled circuits exhibit two-qubit depths ranging from only 8 to 49, while the number of CZ gates varies from 12 to 780. Notably, the growth in two-qubit depth is considerably slower than the increase in total gate count, indicating that the proposed mapping and routing strategy effectively exploits parallel execution of commuting entangling operations while minimizing routing overhead. Consistent with this, the peak entangling metrics do not coincide with the largest circuit by total gate count: the maximum two-qubit depth (49) and CZ count (780) occur on the denser 128-vertex kernels \texttt{sloane\_1zc\_128} and \texttt{sloane\_1dc\_128}, rather than on \texttt{es60fst02}, which has the most gates overall. The hardware interaction parameter $k_{\mathrm{swap}}$ varies between 2 and 12 depending on the connectivity of the reduced graph, enabling efficient embedding without excessive increases in entangling depth.

Overall, these results demonstrate that the combination of aggressive classical graph reduction and hardware-aware quantum circuit synthesis enables the execution of QAOA circuits for MIS kernels containing up to 125 qubits. Since the current IBM quantum processor provides 156 physical qubits, the proposed compilation strategy leaves additional capacity for solving even larger reduced kernels as quantum hardware and gate fidelities continue to improve. The consistently moderate two-qubit depths achieved after transpilation further indicate that the proposed framework is well suited for executing large-scale MIS instances on current-generation superconducting quantum hardware.

\bibliography{aapmsamp}

\end{document}